\documentclass[twoside]{article}

\newif\ifFull

\Fulltrue

\ifFull
\fi

\usepackage[accepted]{aistats2023}

\usepackage{amsfonts}
\usepackage{amssymb}
\usepackage{amsmath}
\usepackage{amsthm}
\usepackage{latexsym}
\usepackage{graphicx}
\usepackage{hyperref}
\usepackage{enumerate}
\usepackage{subcaption}
\usepackage{multirow}
\usepackage{siunitx}

\usepackage{thmtools}
\usepackage{thm-restate}

\usepackage[ruled,linesnumbered]{algorithm2e}
\newcommand{\R}{\mathbb{R}}  
\newcommand{\dd}[1]{\, d#1}  
\newcommand{\T}{\mathrm{T}}  
\newcommand{\Span}{\mathrm{span}}  

\newcommand{\E}{\mathbb{E}}  
\newcommand{\Cov}{\mathrm{Cov}}  
\newcommand{\given}{\mid}  

\newcommand{\M}{\mathcal{M}}  
\newcommand{\data}{\mathfrak{Data}}  

\newcommand{\mup}{\mu}  
\newcommand{\nuq}{\nu}  

\newcommand{\maxdiv}[2]{D_{\infty}(#1\parallel #2)}  
\newcommand{\maxdiva}[3]{D_{\infty}^{#1}(#2\parallel #3)}  
\newcommand{\Winf}[2]{W_{\infty}(#1, #2)}  

\newcommand{\sensW}[2]{\Delta_{W}(#1,#2)}  
\newcommand{\sensE}[2]{\Delta_{E}(#1,#2)}  
\newcommand{\sensEo}[2]{\Delta_{E,1}(#1,#2)}  
\newcommand{\sensEt}[2]{\Delta_{E,2}(#1,#2)}  
\newcommand{\sensEn}[2]{\Delta_{E,n}(#1,#2)}  

\DeclareMathOperator{\Range}{Range}
\DeclareMathOperator{\support}{supp}

\DeclareMathOperator{\Lap}{Lap}
\DeclareMathOperator{\Gauss}{\mathcal{N}}

\declaretheorem[name=Theorem,numberwithin=section]{theorem}

\declaretheorem[name=Corollary,numberwithin=section]{corollary}
\declaretheorem[name=Lemma,numberwithin=section]{lemma}
\declaretheorem[name=Definition,numberwithin=section]{definition}
\declaretheorem[name=Example,numberwithin=section]{example}

\usepackage{natbib}

\begin{document}

\twocolumn[

\aistatstitle{Protecting Global Properties of Datasets with Distribution Privacy Mechanisms}

\aistatsauthor{ Michelle Chen \And Olga Ohrimenko }

\aistatsaddress{ The University of Melbourne \And  The University of Melbourne } ]

\begin{abstract}
We consider the problem of ensuring confidentiality of dataset properties aggregated over many records of a dataset. Such properties can encode sensitive information, such as trade secrets or demographic data, while involving a notion of data protection different to the privacy of individual records typically discussed in the literature. In this work, we demonstrate how a \emph{distribution privacy} framework can be applied to formalize such data confidentiality. We extend the Wasserstein Mechanism from Pufferfish privacy and the Gaussian Mechanism from attribute privacy to this framework, then analyze their underlying data assumptions and how they can be relaxed. We then empirically evaluate the privacy-utility tradeoffs of these mechanisms and apply them against a practical property inference attack which targets global properties of datasets. The results show that our mechanisms can indeed reduce the effectiveness of the attack while providing utility substantially greater than a crude group differential privacy baseline. Our work thus provides groundwork for theoretical mechanisms for protecting global properties of datasets along with their evaluation in practice.
\end{abstract}

\section{INTRODUCTION}
While many notions of privacy have been proposed for protecting individual contributors of data \citep{samarati_generalizing_1998,dinur_revealing_2003,dwork_calibrating_2006,kifer_pufferfish_2014}, there are
 situations where it is instead desirable to protect global properties of a dataset.
A hospital looking to share patient treatment data may, for instance, want to protect overall patient demographics to avoid unfounded claims of correlations between diseases and certain demographics. Similarly, sharing of such data can exacerbate sensitive political issues, as occurred when the 1932 Lebanon census revealed the population's religious makeup. Indeed,  a national census has not been conducted in Lebanon since 1932~\citep{us_international_2019} due to the sensitivity of the matter of religious balance in the region~\citep{10.2307/195924}.

Recent works have developed various \emph{property inference attacks} which demonstrate how data analysis algorithms can leak these global properties of datasets. In contrast to privacy attacks such as membership inference attacks \citep{shokri_membership_2017} and model inversion attacks \citep{fredrikson_model_2015}, these attacks aim to discover  properties aggregated over all records in a dataset rather than properties of individual records. \citet{ateniese_hacking_2015} was the first to formulate such an attack, and further property inference attacks have since been developed against deep neural networks \citep{ganju_property_2018}, large convolutional networks \citep{suri_formalizing_2021}, in federated learning settings~\citep{melis_exploiting_2019}, in black-box settings~\citep{zhang_dataset-level_2020}
and in settings where an attacker can
poison the data~\citep{mahloujifar-poison-2022}.

Currently, there is a lack of rigorous mechanisms for defending against property inference attacks~\citep{suri_formalizing_2021,zhang_attribute_2020}.  Intuitive defenses based on regularization or adding noise \citep{ganju_property_2018,melis_exploiting_2019} lack guarantees of protection. On the other hand, \citet{suri_formalizing_2021} proposed a formal model of property inference attacks as a cryptographic game, but did not propose possible defense mechanisms. 

Similarly, while some privacy frameworks have considered notions of global properties of datasets, they still lack rigorous mechanisms for protecting such properties.
Knowledge hiding \citep{verykios_survey_2008,domadiya_hiding_2013} uses a syntactic approach to  hide sensitive properties aggregated over a dataset, but does not provide privacy guarantees~\citep{kifer_pufferfish_2014}. On the other hand, theoretical frameworks such as Pufferfish privacy~\citep{kifer_pufferfish_2014} and distribution privacy~\citep{kawamoto_local_2019} can model protection of these properties, but have few general mechanisms for this purpose \citep{song_pufferfish_2017,zhang_attribute_2020}. To address this, \citet{zhang_attribute_2020} propose attribute privacy for protecting global properties of datasets. However, their definitions and mechanisms involve restrictive data assumptions, consider queries with a single output and have not been evaluated on real data for their impact on utility or how well they protect against attacks.

We aim to contribute groundwork for a theoretically supported approach for protecting against property inference attacks while overcoming some of these shortcomings of previous works.
Our main contributions are as follows:

\begin{itemize}
    \item We demonstrate how distribution privacy~\citep{kawamoto_local_2019} can be applied to formalize protection of global properties of datasets. This provides an arguably simpler and more general alternative to attribute privacy~\citep{zhang_attribute_2020} as it can capture a wide range of global properties using fewer data assumptions.
    \item We adapt the Wasserstein Mechanism from Pufferfish privacy \citep{song_pufferfish_2017} to distribution privacy. Since this mechanism requires computation of $\infty$-Wasserstein distances which may be large, not well defined or hard to compute in practice, we extend it to allow use of an approximation of $\infty$-Wasserstein distance, with corresponding privacy guarantees. This result is of independent interest and can be applied in other settings where the Wasserstein Mechanism is used (e.g.,~\cite{song_pufferfish_2017,zhang_attribute_2020}).
    \item We propose the Expected Value Mechanism for distribution privacy, which generalizes the Gaussian Mechanism of \citet{zhang_attribute_2020}. In particular, we extend the mechanism to queries with multiple components (e.g., multiple statistics of interest), then analyze directional and adversarial uncertainty assumptions for reducing the noise required by the mechanism.
    \item We evaluate our mechanisms on privacy-utility tradeoffs and effectiveness against a property inference attack. The results show that they can reduce attack accuracy while providing substantially greater utility than a group differential privacy baseline.
    To our knowledge, this is the first evaluation of utility and defense  for mechanisms with theoretical guarantees on protecting global properties of real datasets.
\end{itemize}
\section{PRIVACY FRAMEWORK}

In this section, we formalize the protection of global properties of datasets using the distribution privacy framework~\citep{kawamoto_local_2019}.

\subsection{Differential Privacy}
We first outline differential privacy~\citep{dwork_calibrating_2006}, a widely used privacy framework based on the view that the outcome of a data analysis algorithm does not harm an individual if the algorithm’s
output is the same as if the individual’s data were not used in the
analysis.
\begin{definition}[Differential Privacy]
\label{def:dp}
A mechanism $\M$ satisfies $(\epsilon, \delta)$-differential privacy if for all datasets $D$ and $D'$ which differ on at most one record and all subsets $S\subseteq \Range(\M)$,
\[\Pr(\M(D)\in S)\le \exp(\epsilon)\times \Pr(\M(D')\in S)+\delta.\]
\end{definition}

Differential privacy does not apply directly to our problem, as it protects individual records rather than global properties aggregated over many records. Group differential privacy~\citep{dwork_algorithmic_2014} somewhat
addresses this issue, extending protection to groups of records of data. However group differential privacy mechanisms generally only provide meaningful utility when protecting small groups relative to the size of the whole dataset, making them unsuitable for protecting global properties of datasets.

\subsection{Distribution Privacy}
Distribution privacy \citep{kawamoto_local_2019} is similar to differential privacy, but is defined using the possible underlying distributions of a dataset rather than the presence of individual records.
Formally, the distribution privacy framework is instantiated by specifying a set $\Theta$ of data distributions along with a subset $\Psi\subseteq \Theta\times\Theta$ of pairs of distributions. Each distribution $\theta\in\Theta$ may be viewed as a possible probability distribution of a dataset $\data$, intuitively representing an attacker's beliefs on how $\data$ may have been generated.
The pairs of distributions $\Psi$ then reflect  sensitive properties to be protected from the attacker.
\begin{definition}[Distribution Privacy]
\label{def:distp}
A mechanism $\M$ satisfies $(\epsilon, \delta)$-distribution privacy with respect to a set of distribution pairs $\Psi\subseteq \Theta\times\Theta$ if for all pairs $(\theta_i, \theta_j)\in\Psi$ and all subsets $S\subseteq \Range(\M)$,
\begin{align*}
    &\Pr(\M(\data)\in S\given \theta_i)\\
    &\qquad\le \exp(\epsilon)\times \Pr(\M(\data)\in S\given \theta_j)+\delta, 
\end{align*}
where the expression $\Pr(\M(\data)\in S\given \theta)$ denotes the probability that $\M(\data)\in S$ given $\data\sim \theta$.
\end{definition}
We emphasize that  $\data$ is a random variable rather than a fixed dataset instance.
Thus, the probabilities in the above equation are w.r.t.~randomness from both the mechanism~$\M$ and the distribution of~$\data$. 
In this manner, distribution privacy can be considered to share similarities with noiseless privacy~\citep{bhaskar_noiseless_2011}, Pufferfish privacy \citep{kifer_pufferfish_2014} and other notions of privacy~\citep{bassily_coupledw_2013, machanavajjhala_data_2009} which are defined over possible distributions of a dataset rather than fixed dataset instances.

Using the distribution privacy framework, we can model confidentiality of global properties of a dataset $\data$ by identifying possible underlying distributions of $\data$ given particular values for the sensitive properties of interest. 

\begin{example}[Census Data]
Suppose a government body owns a dataset $\data$ of complete census data and seeks to release a summary $F(\data)$ while protecting sensitive properties related to the population's financial status. The government body could formalize their data confidentiality needs
(e.g., hiding the proportion of low income earners) by first identifying data distributions $\theta_p$ modeling data generation scenarios given each possible proportion $p$ of low income earners. They may then decide that an attacker should not be able to infer the proportion~$p$ to within an absolute error of~$k$. Then distribution privacy is defined w.r.t.~$\Psi$, where each pair of distributions in $\Psi$ takes the form $(\theta_{r\mp d}, \theta_{r\pm d})$ with $|d|\le k$. 
\end{example}

The generality of the distribution privacy framework also allows it to simultaneously capture the protection of multiple properties (global or otherwise) by specifying pairs of data distributions in $\Psi$ appropriately for each sensitive property.

Identifying appropriate data distributions $\Theta$ and pairs $\Psi$ is an important
data modeling step
when using the distribution privacy framework. For our experiments in Section~\ref{sec:exps}, we assume that aggregated Census data can be approximated by some multivariate Gaussian distribution and use data sampling to estimate possible parameters. In general, we take an approach similar to Pufferfish privacy \citep{kifer_pufferfish_2014} and assume the existence of a domain expert who can instantiate the framework according to the application requirements. 

\subsection{Relation to Other Privacy Frameworks}

While we adopt the distribution privacy definition of \citet{kawamoto_local_2019}, our work applies their definition in a different context. \citet{kawamoto_local_2019} focused on protecting the distribution of individual records for location based services, corresponding to the case where $\data$ is a single record of data. In contrast,  we consider $\data$ to be an entire dataset, in which case the Tupling Mechanism of \cite{kawamoto_local_2019} and other distribution privacy mechanisms are not applicable.

In this context, the distribution privacy framework  can be viewed as a generalization of the  distributional attribute privacy framework proposed by~\citet{zhang_attribute_2020} for protecting sensitive attributes aggregated across a dataset. However, while distributional attribute privacy requires conditional marginal distributions to be known for all attributes of a dataset, our privacy framework requires such distributions to be known only for the sensitive properties of interest. Distribution privacy thus has the capacity to capture a wider range of global properties while requiring fewer data assumptions.

\section{WASSERSTEIN MECHANISM}

\label{sec:Wass}
In this section, we adapt the Wasserstein Mechanism from Pufferfish privacy \citep{song_pufferfish_2017} to  distribution privacy. Since this mechanism requires computation of $\infty$-Wasserstein distances which may be large or not well defined in practice, we also propose a variation, the Approximate Wasserstein Mechanism, which can be applied in more general settings. We note that our analysis is of independent interest as it can also be used by other mechanisms that require~$\infty$-Wasserstein distances.

For an instantiation of distribution privacy  specified by a set of pairs of distributions $\Psi\subseteq\Theta\times\Theta$,  let $F(\data)$ denote a query, or function of the data which the data curator would like to release. 
We assume that $F$ takes values in $\R^m$.
Letting $f_\theta$ denote the distribution of $F(\data)$ given $\data\sim\theta$, we intuitively need to apply enough noise to prevent an attacker from determining whether $F(\data)$ was drawn from $f_{\theta_i}$ or $f_{\theta_j}$, for each pair of distributions $(\theta_i,\theta_j)\in\Psi$.
For Pufferfish privacy, \citet{song_pufferfish_2017} identify $\infty$-Wasserstein distance as one suitable measure for determining the amount of noise required.

\begin{definition}[$\infty$-Wasserstein Distance]
\label{def:wasserstein_real}
Let $\mup$ and $\nuq$ be two  distributions on $\R^m$. Let $\Gamma(\mup, \nuq)$ be the set of all joint distributions with marginals $\mup$ and $\nuq$. The $\infty$-Wasserstein distance $\Winf{\mup}{\nuq}$ between $\mup$ and $\nuq$ is defined as
\begin{align*}
    \Winf{\mup}{\nuq}=\inf_{\gamma\in\Gamma(\mup,\nuq)}\max_{(x,y)\in \support(\gamma)}\|x-y\|_1.
\end{align*}
\end{definition}

Each $\gamma\in\Gamma(\mup,\nuq)$ may be interpreted as a way of transforming $\mup$ into $\nuq$ by shifting probability mass
between the two distributions. The expression $\max_{(x,y)\in \support(\gamma)}\|x-y\|_1$ can then be interpreted as the cost of $\gamma$, representing the maximum $L_1$ distance traveled by a probability mass in the shifting described by $\gamma$. Thus, the $\infty$-Wasserstein distance between distributions $\mup$ and $\nuq$ is intuitively the maximum distance traveled by a probability mass when transforming $\mup$ into $\nuq$ in the most cost-efficient manner possible.

We adapt the Wasserstein Mechanism of \citet{song_pufferfish_2017} by scaling Laplace noise  to the maximum $\infty$-Wasserstein distance between  pairs of distributions of interest
\begin{align*}
    \sensW{\Psi}{F}=\sup_{(\theta_{i}, \theta_{j})\in \Psi}\Winf{f_{\theta_i}}{ f_{\theta_j}}.
\end{align*}

\begin{restatable}{thm}{wasserstein}
\label{thm:wasserstein}
The mechanism $\M(\data)=F(\data)+Z$, where $Z\sim\Lap(0, (\sensW{\Psi}{F}/\epsilon) I)$, satisfies $(\epsilon,0)$-distribution privacy with respect to $\Psi$.
\end{restatable}

Now, this mechanism assumes that $\Winf{f_{\theta_i}}{ f_{\theta_j}}$ is well-defined for each pair $(\theta_i,\theta_j)\in\Psi$. This can be an issue if, for instance, the query function $F$ is unbounded. Similarly, if $F$ takes on a large range of values with low probability (e.g., average salary), these $\infty$-Wasserstein distances can be large, resulting in more noise than is necessary for the majority of cases. Thus, we prove an variation of the Wasserstein Mechanism which can require less noise in exchange for a small probability of  loss of privacy.

\subsection{Approximate Wasserstein Mechanism}
\label{sec:wasserstein_approx}
Our variation essentially allows a low probability set of possible values of $F$ to be disregarded while still providing  $(\epsilon,\delta)$-distribution privacy. Formally, we first generalize $\infty$-Wasserstein distance as follows.

\begin{definition}[$(W, \delta)$-closeness]
\label{def:closeness}
Let $\mup$ and $\nuq$ be two  distributions on $\R^m$. We say that $\mup$ and $\nuq$ are \emph{$(W, \delta)$-close}  if there exists a distribution $\gamma\in\Gamma(\mup,\nuq)$ and subset $R\subseteq \support(\gamma)$ such that
\begin{align*}
    \|x-y\|_{1}\le W\qquad\forall \, (x,y)\in R
\end{align*}
and
\begin{align*}
    \int\int_{(x,y)\in R}\gamma(x,y)\dd{x}\dd{y}\ge 1-\delta.
\end{align*}
\end{definition}

Recalling the view of each $\gamma\in\Gamma(\mup,\nuq)$ as a transformation of probability distributions, we may interpret $(W,\delta)$-closeness as the ability to transform one distribution into another by shifting each probability mass by a distance of at most $W$, with exceptions of mass at most $\delta$.

The parameter $\delta$ allows for two distributions $\mup$ and $\nuq$ to be considered $(W,\delta)$-close for values of $W$ that are potentially much smaller than $\Winf{\mup}{\nuq}$. To illustrate, consider distributions $\mup$ and $\nuq$ 
such that $\mup$ assigns probabilities $0.6$, $0.2$, $0$, $0.2$ and $\nuq$ assigns probabilities $0.4$, $0.3$,  $0.2$, $0.1$ to elements $1$, $2$,  $3$, $100$ respectively. Then, as shown in Figure \ref{fig:close_distributons},  $\Winf{\mup}{\nuq}=97$ is large since probability mass must be shifted from $\mup(100)$ to $\nuq(3)$. However, by disregarding the shift associated with this particular probability mass, we see that $\mup$ and $\nuq$ are $(W,0.1)$-close with $W=1$.

\begin{figure}
\begin{center}
\includegraphics{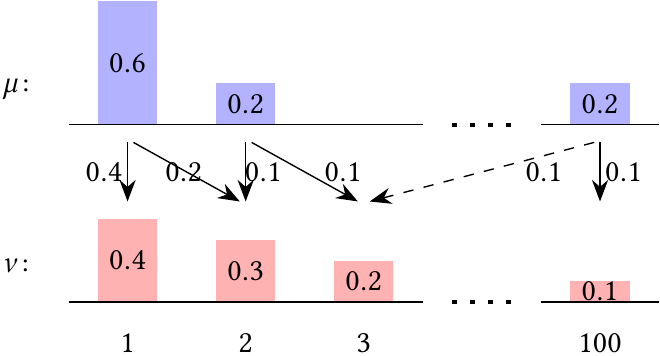}
\vspace{.2in}
\caption{Distinction between $\infty$-Wasserstein distance and $(W,\delta)$-closeness. $\Winf{\mup}{\nuq}=97$ due to the shift represented by the dashed arrow, while  disregarding this shift shows that $\mup$ and $\nuq$ are $(W,0.1)$-close with $W=1$.
}
\label{fig:close_distributons}
\end{center}
\end{figure}

This difference shows how scaling noise according to $(W,\delta)$-closeness can result in significantly less noise than  scaling to $\infty$-Wasserstein distance. The ability to disregard a small mass $\delta$ also allows $(W,\delta)$-closeness to be well-defined even when $\infty$-Wasserstein distance may not be.

\begin{restatable}{thm}{wassersteinapprox}
\label{thm:wasserstein_approx}
Suppose that, for all pairs $(\theta_i,\theta_j)\in\Psi$, the distributions $f_{\theta_{i}}$ and $f_{\theta_j}$ are $(W,\delta)$-close. Then, the mechanism $\M(\data)=F(\data)+Z$, where $Z\sim\Lap(0, (W/\epsilon) I)$, satisfies $(\epsilon,\delta)$-distribution privacy with respect to $\Psi$.
\end{restatable}
\begin{proof}
Let $(\theta_i,\theta_j)\in\Psi$ be a pair of distributions. Since $f_{\theta_i}$ and $f_{\theta_j}$ are $(W,\delta)$-close, there exists a distribution $\gamma\in\Gamma(f_{\theta_i},f_{\theta_j})$  and a subset $R\subseteq\support(\gamma)$ as in the conditions of Definition \ref{def:closeness}. Now, for all $S\subseteq \Range(\M)$,
\begin{align*}
&\Pr(\M(\data)\in S\given\theta_i)\\
&=\int_{t}\Pr(F(\data)=t\given\theta_i)\Pr(Z+t\in S)\dd{t}\nonumber\\
&=\int_{t}\int_{s}\gamma(t,s)\Pr(Z+t\in S)\dd{s}\dd{t}.
\end{align*}
Since each component of noise $Z_k\sim\Lap(W/\epsilon)$ is independently sampled, the Laplace Mechanism from differential privacy \citep{dwork_calibrating_2006} 
shows that
\begin{align*}
&\int\int_{(t,s)\in R}\gamma(t,s)\Pr(Z+t\in S)\dd{s}\dd{t}\nonumber\\
&\le \exp(\epsilon)\int\int_{(t,s)\in R}\gamma(t,s)\Pr(Z+s\in S)\dd{s}\dd{t},
\end{align*}
since $\|t-s\|_{1}\le W$ for all $(t,s)\in R$. 
We also have
\begin{align*}
&\int\int_{(t,s)\not\in R}\gamma(t,s)\Pr(Z+t\in S)\dd{s}\dd{t}\\
&\le \int\int_{(t,s)\not\in R}\gamma(t,s)\dd{s}\dd{t}\le\delta.
\end{align*}
Thus,
\begin{align*}
&\Pr(\M(\data)\in S\given\theta_i)\\
&\le \exp(\epsilon)\int\int_{(t,s)\in R}\gamma(t,s)\Pr(Z+s\in S)\dd{s}\dd{t}+\delta\\
&\le\exp(\epsilon)\Pr(\M(\data)\in S\given\theta_j)+\delta.
\end{align*}
It follows that $\M$ satisfies $(\epsilon,\delta)$-distribution privacy.
\end{proof}

We remark that $(W,\delta)$-closeness can be viewed as an approximation of the $\infty$-Wasserstein distance between two distributions using a subset of their supports. 
More precisely, given distributions $\mup$ and $\nuq$, suppose that there exist subsets $R_1\subseteq \support(\mup)$ and $R_2\subseteq \support(\nuq)$ such that $\Pr(\mup\in R_1)\ge 1-\delta/2$ and $\Pr(\nuq\in R_2)\ge 1-\delta/2$. The union bound can then  be used to show that   $\Pr(\gamma\in R_1\times R_2)\ge 1-\delta$ for every coupling $\gamma$ of $\mup$ and $\nuq$. Hence,  $\mup$ and $\nuq$ are $(W,\delta)$-close, where 
\begin{align*}
    W=\inf_{\gamma\in\Gamma(\mup,\nuq)}\max_{(x,y)\in \support(\gamma)\cap (R_1\times R_2)}\|x-y\|_1.
\end{align*}
Comparing with the definition of $\infty$-Wasserstein distance, the above expression can be interpreted as only considering the movement of probability masses between the subsets $R_1\subseteq \support(\mup)$ and $R_2\subseteq \support(\nuq)$. Thus, $(W,\delta)$-closeness formalizes the notion that $\Winf{\mup}{\nuq}$ can be approximated using subsets of the supports of $\mup$ and $\nuq$. In turn, this provides privacy guarantees when applying the Wasserstein Mechanism with common methods of computing $\infty$-Wasserstein distance, which  proceed by approximating continuous distributions with bounded or discrete distributions \citep{panaretos_statistical_2019}.

\subsection{Functions Bounded with High Probability}
\label{sec:wasserstein_bounded}

As an example of using the Approximate Wasserstein Mechanism, consider settings in which the queries of interest are bounded with high probability, as is often the case for aggregated statistics such as average salary. Such settings allow the following efficient application of the Approximate Wasserstein Mechanism.

\begin{restatable}{thm}{boundedwhp}
\label{prop:bounded_whp}
    Suppose that $\|F(\data)-\E[F(\data)]\|_1\le c$ with probability at least $1-\delta/2$ for each $\data\sim\theta\in\Theta$. Let 
\begin{align*}
    \sensE{\Psi}{F}=\sup_{(\theta_{i},\theta_{j})\in\Psi} \|\E[f_{\theta_i}]- \E[f_{\theta_j}]\|_1.
\end{align*}
    Then, the distributions $f_{\theta_i}$ and $f_{\theta_j}$ are $(\sensE{\Psi}{F}+2c,\delta)$-close for all pairs $(\theta_i,\theta_j)\in\Psi$.
\end{restatable}

The proof of Proposition \ref{prop:bounded_whp} is provided in Appendix \ref{ap:wass_proofs}.

Given Proposition \ref{prop:bounded_whp}, we can apply
Theorem \ref{thm:wasserstein_approx} to deduce that Laplace noise with scale  $(\sensE{\Psi}{F}+2c)/\epsilon$ is sufficient for $(\epsilon, \delta)$-distribution privacy with respect to $\Psi$. Thus, if $F$ is bounded by a small range of values with high probability, then $c$ will be small and  the Approximate Wasserstein Mechanism can be applied while requiring only a small amount of noise.

Observe also that $\sensE{\Psi}{F}+2c$ can be well-defined even if $F$ is unbounded, as long as $F$ is bounded with high probability. Moreover, it only requires expected values and bounds for $f_{\theta}$ to be computed. This is in contrast to the  $\infty$-Wasserstein distances required by the Wasserstein Mechanism (Theorem \ref{thm:wasserstein}), which can be  expensive to compute \citep{song_pufferfish_2017} and moreover may not be well-defined if $F$ is unbounded.

\section{EXPECTED VALUE MECHANISM}
\label{sec:exp_all}
While theoretically applicable to general instantiations of distribution privacy, the Wasserstein Mechanism can be too computationally expensive to use in practice due to the need for $\infty$-Wasserstein distances \citep{song_pufferfish_2017}. The variation we proposed can accommodate more efficient approximations, however, in general, $(W,\delta)$-closeness may still be difficult to compute.

We now propose the Expected Value Mechanism, an alternative, efficient mechanism that can achieve distribution privacy in specialized settings. In particular, it can be applied when the pairs of distributions to be protected are approximate translations of each other (e.g., Gaussian with similar variance). 
Our base mechanism adapts the Gaussian Mechanism from attribute privacy \citep{zhang_attribute_2020}.
We then propose two variants that have smaller noise requirements in circumstances when
directional and adversarial uncertainty assumptions can be made, allowing insights of a domain expert to be incorporated to improve utility.
Finally, we formalize how data assumptions can be relaxed while still providing confidentiality guarantees.
Note that theorem proofs
are provided in the appendix.

\subsection{Mechanism}
\label{sec:basic_exp}
As in previous sections, let $f_\theta$ denote the distribution of $F(\data)$ given  $\data\sim \theta$. We assume that $f_{\theta_i}$ and $f_{\theta_j}$ are translations of each other for each pair $(\theta_i,\theta_j)\in\Psi$. That is, for each such pair, there exists $c\in\R^m$ such that $f_{\theta_i}(t)=f_{\theta_j}(t+c)$ for all $t\in\R^m$. This occurs, for instance, when each distribution can be assumed to be Gaussian with the same variance but different means.

The Expected Value Mechanism achieves distribution privacy by applying Laplace or Gaussian noise proportional to the worst case distance between expected values
\begin{align*}
    \sensEn{\Psi}{F}=\sup_{(\theta_{i},\theta_{j})\in\Psi} \|\E[f_{\theta_i}]- \E[f_{\theta_j}]\|_n.
\end{align*}

\begin{restatable}{thm}{explap}
\label{thm:exp_lap}
Suppose that $f_{\theta_i}$ is a translation of $f_{\theta_j}$ for every pair $(\theta_i,\theta_j)\in\Psi$.
Then, the mechanism $\M(\data)=F(\data)+Z$, where $Z\sim\Lap(0, (\sensEo{\Psi}{F}/\epsilon) I)$, satisfies $(\epsilon,0)$-distribution privacy with respect to $\Psi$.
\end{restatable}

\begin{restatable}{thm}{expgauss}
\label{thm:exp_gauss}
Suppose that $f_{\theta_i}$ is a translation of $f_{\theta_j}$ for every pair $(\theta_i,\theta_j)\in\Psi$.
Let $\sigma\ge c\sensEt{\Psi}{F}/\epsilon$, where $\epsilon\in(0,1)$ and $c=\sqrt{2\ln (1.25/\delta)}$ for some $\delta>0$.
Then, the mechanism $\M(\data)=F(\data)+Z$, where $Z\sim\Gauss(0, \sigma^2 I)$, satisfies $(\epsilon,\delta)$-distribution privacy with respect to $\Psi$.
\end{restatable}

\subsection{Variant Using Directional Assumptions}
\label{sec:dir_exp}
We now propose a variant based on the insight that distribution privacy may sometimes be achievable while applying noise only to particular directions or components for a query.
To illustrate, consider a company seeking to release customer statistics while protecting their proportion of female customers. Suppose that particular statistics, such as customer income and investments, tend to increase by certain amounts as the proportion $p$ of female customers varies. If the company models these statistics as changing in an approximately constant direction $v$ as the proportion $p$ changes, it may suffice to add noise only in the direction $v$ to ensure protection of the global property $p$. We formalize this with the Directional Expected Value Mechanism, described in Theorem \ref{thm:dir_exp} as follows.

\begin{restatable}{thm}{direxp}
\label{thm:dir_exp}
Suppose that $f_{\theta_i}$ is a translation of $f_{\theta_j}$ and furthermore that the vector $\E[f_{\theta_i}]-\E[f_{\theta_j}]$ is parallel to the unit vector $v$ for each $(\theta_i,\theta_j)\in \Psi$. Then, the mechanism $\M(\data)=F(\data) + Yv$, where $Y\sim\Lap(\sensEt{\Psi}{F}/\epsilon)$, satisfies $(\epsilon,0)$-distribution privacy with respect to $\Psi$. 
\end{restatable}

If the differences of mean vectors $\E[f_{\theta_i}]-\E[f_{\theta_j}]$ are not all parallel to some vector $v$, the mechanism can be modified by first identifying a set of orthogonal vectors $v_1, v_2,\dots, v_d$ which span the possible directions of $\E[f_{\theta_i}]-\E[f_{\theta_j}]$ as the pairs $(\theta_i,\theta_j)\in\Psi$ vary. Noise proportional to $\sensEt{\Psi}{F}/\epsilon$ can then be applied in the directions of $v_1, v_2,\dots, v_d$, reducing the overall noise requirement if the number of directions $d$ is small in comparison to the total dimension $m$ of the query.
Alternatively, approximations $\tilde{f}_{\theta}$ can be chosen to satisfy the directional assumptions and used in place of the exact distributions~$f_{\theta}$. In this case, analysis can proceed as for our general results on relaxing mechanism assumptions, discussed later in Section~\ref{sec:relax}.
Note also that the mechanism can readily be modified to use Gaussian noise instead of Laplace noise.
\subsection{Variant Using Adversarial Uncertainty}
\label{sec:exp_gauss_variants}
We now discuss how adversarial uncertainty
in the data distributions can be exploited by the mechanism. Using the scenario from the previous section as an example, variance in income may be enough to prevent accurate estimation of the proportion of female customers even if there is a correlation between these properties. The company may then be able to safely release their average customer income with little to no applied noise. 

For this variant, we assume the distribution $f_\theta$ of the query function to be multivariate Gaussian with mean $\mu_\theta$ and covariance matrix $\Sigma_{\theta}$ for each $\theta\in\Theta$. We also assume that  $f_{\theta_i}$ and $f_{\theta_j}$ have the same covariance matrix $\Sigma_{\theta_i}=\Sigma_{\theta_j}$ for each pair $(\theta_i,\theta_j)\in\Psi$.
With these assumptions, we first identify a set of conditions under which a query $F(\data)$ may be safely released without additional noise. 

\begin{restatable}{thm}{gaussreducedmulti}
\label{thm:gauss_reduced_multi}
Suppose that $f_{\theta}\sim\Gauss(\mu_{\theta}, \Sigma_\theta)$  for each $\theta\in\Theta$ and that $\Sigma_{\theta_i}=\Sigma_{\theta_j}$ for each $(\theta_i,\theta_j)\in\Psi$.
Let $c=\sqrt{2\ln(1.25/\delta)}$. Then, the mechanism $\M(\data)=F(\data)$ satisfies $(\epsilon,\delta)$-distribution privacy with respect to $\Psi$ as long as
\begin{align*}
    (\mu_{\theta_i}-\mu_{\theta_j})^{\T} \Sigma_{\theta_i}^{-1}(\mu_{\theta_i}-\mu_{\theta_j})\le (\epsilon/c)^2
\end{align*}
for all $(\theta_i,\theta_j)\in\Psi$.
\end{restatable}

The key observation underlying Theorem \ref{thm:gauss_reduced_multi} is that every multivariate Gaussian variable can be expressed as a transformation $X=AZ+\mu$ of a standard normal  vector $Z$, and $\Pr(X=x)=\frac{1}{|A|}\Pr(Z=A^{-1}x)$ for all $x\in\R^m$.
The proof then proceeds similarly to that of the Gaussian Mechanism from differential privacy \citep{dwork_algorithmic_2014}.

From Theorem \ref{thm:gauss_reduced_multi}, it follows that privacy can be ensured when there is sufficient noise in the eigenvector directions of each covariance matrix $\Sigma_{\theta}$. When these eigenvector directions are the same for each $\theta\in\Theta$, we thus arrive at the Eigenvector Gaussian Mechanism as follows.
\begin{restatable}{thm}{gausseig}
\label{thm:gauss_eig}
Suppose that $f_{\theta}\sim\Gauss(\mu_{\theta}, \Sigma_\theta)$  for each $\theta\in\Theta$ and that $\Sigma_{\theta_i}=\Sigma_{\theta_j}$ for each $(\theta_i,\theta_j)\in\Psi$. Suppose furthermore that the normalized eigenvectors $v_1, v_2,\dots, v_m$ of  $\Sigma_{\theta}$ are the same for each $\theta\in\Theta$. Let $c=\sqrt{2\ln(1.25/\delta)}$. For each $\theta\in\Theta$, let
\begin{align*}
    \sigma_{\theta,k}^2&=\max\left(0, (c\sensEt{\Psi}{F}/\epsilon)^2 - \lambda_{\theta,k}^2\right)
\end{align*}
where $\lambda_{\theta,k}^2=v_k^{\T}\Sigma_{\theta}v_k$ are squared eigenvalues. Let $\sigma_k^2=\max_{\theta\in\Theta}\sigma_{\theta,k}^2$ and $\Sigma=\sum_{k}\sigma_k^2v_{k}v_{k}^{\T}$. Then, the mechanism $\M(\data)=F(\data)+Z$, where $Z\sim\Gauss(0,\Sigma)$, satisfies $(\epsilon,\delta)$-distribution privacy with respect to $\Psi$. 
\end{restatable}

\paragraph{Example:} 
Consider an instantiation of distribution privacy given by  $\Theta=\{\theta_1,\theta_2\}$ and  $\Psi=\{(\theta_1,\theta_2),(\theta_2,\theta_1)\}$ and suppose that $f_{\theta_1}$ and $f_{\theta_2}$ are normally distributed with
$$\mu_{\theta_1} = 
\begin{bmatrix}
100\\
101
\end{bmatrix},
\,
\mu_{\theta_2} = 
\begin{bmatrix}
99\\
102
\end{bmatrix},
\,
\Sigma_{\theta_1}=\Sigma_{\theta_2}=
\begin{bmatrix}
22&-6\\
-6&13
\end{bmatrix}.
$$
If $\epsilon=1$ and $\delta=0.001$, then $(c\sensEt{\Psi}{F}/\epsilon)^2\approx 28.52$. Since the eigenvectors of $\Sigma_{\theta_1}=\Sigma_{\theta_2}$ are  $v_1=(1/\sqrt{5},2/\sqrt{5})$ and $v_2=(2/\sqrt{5},-1/\sqrt{5})$ with corresponding eigenvalues 10 and 25, the Eigenvector Gaussian Mechanism is equivalent to adding Gaussian noise with variance $\sigma_{1}^2\approx 18.52$ in the direction of $v_1$ and variance $\sigma_{2}^2\approx 3.52$ in the direction of $v_2$. In comparison, if adversarial uncertainty was not exploited, Gaussian noise with variance $28.52$ would need to be added in both directions --- significantly impacting the utility.

\medskip

We note that the Gaussian Mechanism of \citet{zhang_attribute_2020} also exploits adversarial uncertainty, but only in the one-dimensional case. Thus, our work generalizes this idea to higher dimensions. We also remark that adversarial uncertainty can be exploited together with the directional assumptions of Section \ref{sec:dir_exp}, as we detail in Appendix \ref{ap:dir_gauss}.
\subsection{Privacy Guarantees Under Relaxed Assumptions}
\label{sec:relax}
In our analysis so far, we adopted various translation assumptions, directional assumptions, and adversarial uncertainty assumptions on the distributions involved.
These assumptions provide an initial means of incorporating domain knowledge to improve utility, however, may be difficult to apply in practice.
One approach to relax the assumption requirements is to use approximations $\tilde{f}_{\theta}$ in place of the exact distributions $f_{\theta}$, producing a mechanism $\widetilde{\M}(\data)=F(\data)+\widetilde{Z}$ where the query $F(\data)$ follows a true distribution $f_\theta$ for some $\theta\in\Theta$, but $\widetilde{Z}$ is computed based on the approximations $\tilde{f}_{\theta}$. The resultant privacy loss can then be quantified with max-divergence, similar to~\citet{zhang_attribute_2020} for their setting.
\begin{definition}[$\delta$-Approximate Max-Divergence]
\label{def:maxdiv_approx}
Let $\mup$ and $\nuq$ be two distributions. The $\delta$-approximate max-divergence $\maxdiva{\delta}{\mup}{\nuq}$ between $\mup$ and $\nuq$ is defined as
\begin{align*}
    \maxdiva{\delta}{\mup}{\nuq}=\sup_{S\subseteq\support(\mup):\Pr(\mup\in S)\ge \delta}\ln\frac{\Pr(\mu\in S)-\delta}{\Pr(\nuq\in S)}.
\end{align*}
\end{definition}

\begin{restatable}{thm}{othermaxdiva}
\label{thm:other_maxdiva}
Suppose that $\M$ satisfies $(\epsilon,\delta)$-distribution privacy with respect to $\Psi$ and that $\M(\data)\given F(\data)$ is independent of $\data$. For each $\theta\in\Theta$, let $\tilde{f}_{\theta}$ be an approximation of $f_{\theta}$ and suppose that $$\max\left(\maxdiva{\eta}{\tilde{f}_{\theta}}{f_{\theta}},\maxdiva{\eta}{f_{\theta}}{\tilde{f}_{\theta}}\right)\le \lambda$$
for all $\theta\in\Theta$. Then, the mechanism $\widetilde{\M}$ which uses the approximations $\tilde{f}_{\theta}$ in place of the true distributions $f_{\theta}$ satisfies $(\epsilon',\delta')$-distribution privacy with respect to $\Psi$, where $\epsilon'=\epsilon+2\lambda$ and $\delta'=(1+\exp(\epsilon+\lambda))\eta+\exp(\lambda)\delta$. 
\end{restatable}

We furthermore observe that the privacy loss can alternatively be controlled with a small amount of additional noise, using $\infty$-Wasserstein distance to quantify possible deviations between the assumed approximations and true distributions. 
\begin{restatable}{thm}{relaxwaslap}
\label{thm:relax_was_lap}
Suppose that the mechanism $\M(\data)=F(\data)+Z$  satisfies $(\epsilon,\delta)$-distribution privacy with respect to $\Psi$. For each $\theta\in\Theta$, let $\tilde{f}_{\theta}$ be an approximation of $f_{\theta}$ and let
$$W=\sup_{\theta\in\Theta}\Winf{f_{\theta}}{\tilde{f}_{\theta}}.$$
Let $Z'=(Z'_1,Z'_2,\dots, Z'_m)$, where $Z'_k\sim \Lap(W/\lambda)$ for each $k$ and let $\widetilde{\M}$ be the mechanism which applies $\M$ using the approximations $\tilde{f}_{\theta}$ in place of the true distributions $f_{\theta}$. Then the mechanism which outputs  $\widetilde{\M}(\data)+Z'$ satisfies  $(\epsilon',\delta')$-distribution privacy with respect to $\Psi$, where $\epsilon'=\epsilon+2\lambda$ and $\delta'=\exp(\lambda)\delta$. 
\end{restatable}

These results quantify how much protection can be ensured if modeling assumptions made by a domain expert do not exactly hold. Thus, they provide justification for the use of these assumptions when applying our mechanisms by providing fallback guarantees where necessary.

\begin{figure*}
    \centering
    \begin{subfigure}[t]{0.33\textwidth}
        \centering
        \includegraphics{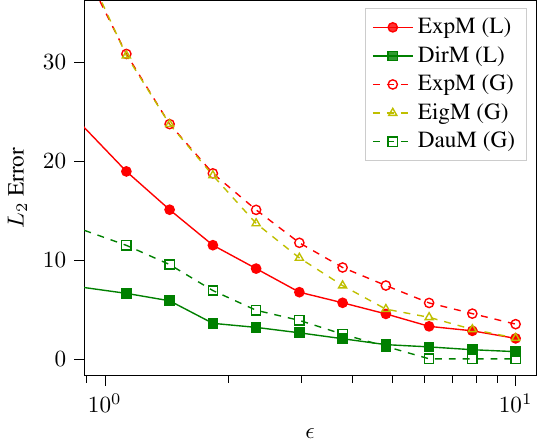}
        \label{fig:a_util_mech_b}
        \caption{Income, $\Delta p_I=0.1$.}
    \end{subfigure}
    \hfill
    \begin{subfigure}[t]{0.33\textwidth}
        \centering
        \includegraphics{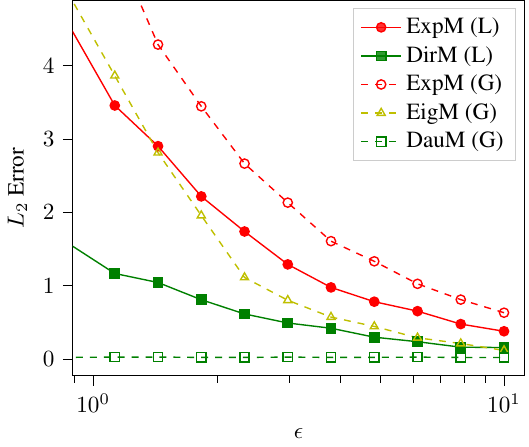}
        \label{fig:a_util_mech_d}
         \caption{Work class, $\Delta p_W=0.04$.}
    \end{subfigure}
    \hfill
    \begin{subfigure}[t]{0.33\textwidth}
        \centering
        \includegraphics{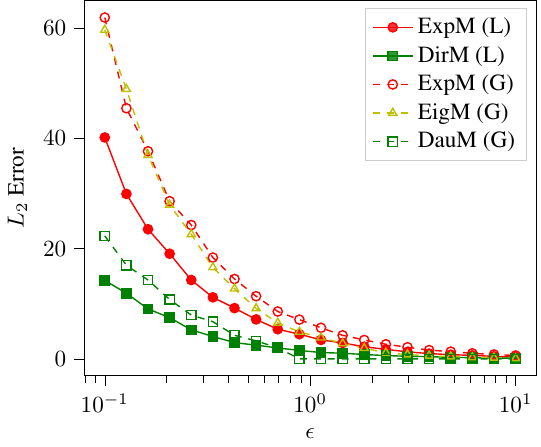}
        \label{fig:a_util_mech_c}
         \caption{Work class, $\Delta p_W=0.04$.}
    \end{subfigure}
    \caption{Privacy-utility tradeoffs of our mechanisms measured by $L_2$ error when (a) protecting income, $\Delta p_I=0.1$; and (b), (c) protecting work class, $\Delta p_W=0.04$, for varying $\epsilon$. For the Gaussian mechanisms, $\delta=0.001$.
    }
    \label{fig:a_util_mech}
\end{figure*}

\begin{figure*}[t]
    \centering
    \begin{subfigure}[t]{0.48\textwidth}
        \centering
        \includegraphics{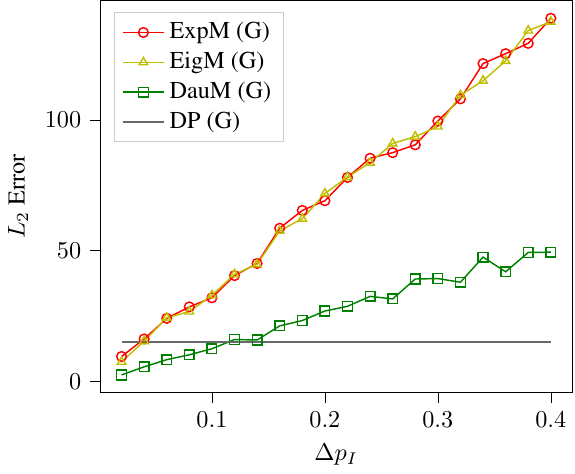}
        \caption{Income property $p_I$.}
        \vspace{.2in}
        \label{fig:a_util_pdiff_b}
    \end{subfigure}
    \begin{subfigure}[t]{0.48\textwidth}
        \centering
        \includegraphics{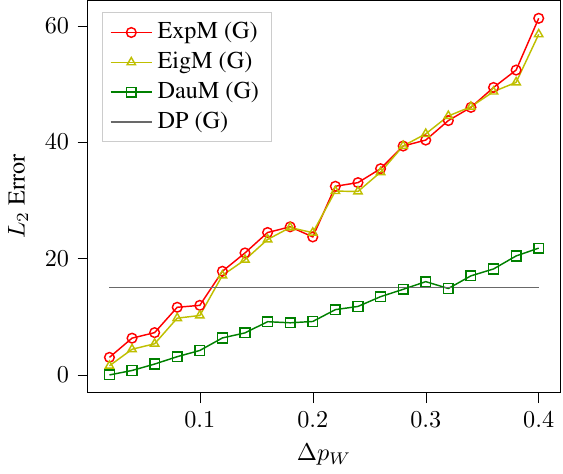}
        \caption{Work class property $p_W$.}
        \label{fig:a_util_pdiff_d}
    \end{subfigure}
    \caption{Error incurred by our mechanisms as $\Delta p$ varies, in comparison to the Gaussian Mechanism from differential privacy (DP). Here, $\epsilon=1$ and $\delta=0.001$.}
    \label{fig:a_util_pdiff}
\end{figure*}

\section{EXPERIMENTS}
\label{sec:exps}
We now evaluate our mechanisms on privacy-utility tradeoffs
and effectiveness against property inference attacks.\footnote{Our code for these experiments can be found at \url{https://github.com/mgcsls/mechanisms-global-properties}}

\subsection{Methodology}
We use the Adult dataset from the UCI Machine Learning repository \citep{dua_uci_2017}, which consists of records from the 1994 US Census database. We remove records with missing attributes, then for the remaining \num{45222} records we randomly allocate \num{10000}  as auxiliary data to implement attacks, \num{10000}  as testing data to evaluate the accuracy of the attacks, and the remainder for modeling data distributions to implement our mechanisms.

Our task is to release several statistics on a subset of 100 records, namely average age, average number of years of education, number of individuals who have never married, number of female individuals, and average hours worked per week.
We consider two possible sensitive properties global for a particular subset 
--- the proportion~${p_I}$ of individuals in the subset whose income exceeds~$\$50$K, and the proportion~${p_W}$ who are private sector workers.
For each, we want to prevent an attacker from determining whether $p$ has value $0.5-\Delta p/2$ or $0.5+\Delta p/2$, for $\Delta p \in [0.02, 0.4]$.

We evaluate five variants of the Expected Value Mechanism, summarized in Table \ref{tab:exp_mechanisms}. In order to apply these mechanisms, with the statistics of interest as the query $F$, we assume that the possible query distributions $f_{\theta}$ are approximately multivariate Gaussian with a  shared covariance matrix $\Sigma_{\theta_i}=\Sigma_{\theta_j}$ for all pairs $(\theta_i,\theta_j)\in\Psi$. We compute appropriate mean vectors $\mu_{\theta}$ and covariance matrices $\Sigma_{\theta}$ by randomly sampling 1000 subsets of 100 records from the Adult dataset with global sensitive property values as specified by $\theta$, with details in Appendix \ref{ap:more_results_modeling}.

We note that while these modeling steps may be difficult in practice, here they nonetheless allow initial evaluation of the utility and defense provided by our distribution privacy approach.

\begin{table}
\caption{
Mechanisms Implemented in the Evaluation.
}
\label{tab:exp_mechanisms}
\begin{center}
\begin{tabular}{ll}
\hline
\textbf{Variant of Expected Value Mechanism}                    & \textbf{Shorthand} \\ \hline \hline
Laplace (Section \ref{sec:basic_exp})  
& ExpM (L)        \\ \hline
Laplace, Directional (Section \ref{sec:dir_exp})
& DirM (L)         \\ \hline
Gaussian (Section \ref{sec:basic_exp})   
& ExpM (G)         \\ \hline
Gaussian, Eigenvector (Section \ref{sec:exp_gauss_variants})         
& EigM (G)         \\ \hline

\begin{tabular}[c]{@{}l@{}}Gaussian, Directional with\\ Adversarial Uncertainty (Section \ref{sec:exp_gauss_variants})  \end{tabular}         
& DauM (G)         \\ \hline
\end{tabular}
\end{center}
\end{table}

As a baseline, we use the Laplace and Gaussian Mechanisms from group differential privacy \citep{dwork_algorithmic_2014} applied with group size $k=n=100$, the size of the subsets queried. We note also that the Wasserstein Mechanism and Gaussian Mechanism of \citet{zhang_attribute_2020} are equivalent to ExpM (L) and ExpM (G) under our modeling assumptions. Thus, our experiments also provide insight into the practical impact of these mechanisms, which have not been evaluated on real datasets until now.

\subsection{Privacy-Utility Tradeoffs}
\label{sec:a_priv_util}
Figure \ref{fig:a_util_mech} shows privacy-utility tradeoffs as our mechanisms are applied to protect the income and work class properties to additive factors $\Delta p_I=0.1$ and $\Delta p_W=0.04$, respectively. Group differential privacy baselines are not included in this figure since they incurred error values an order of magnitude larger than the Expected Value Mechanism, as shown in Table \ref{tab:a_util} for the Gaussian variants. Note that all error values are averaged over 50 repetitions of sampling subsets of the Adult dataset, computing the statistics of interest, then applying each mechanism. 

\begin{table}
\caption{$L_2$ Error Incurred Protecting the Income Property to an Additive Factor of $\Delta p_I=0.1$. Here, $\delta=0.001$.}
\label{tab:a_util}
\begin{center}
\begin{tabular}{lccc}
\cline{2-4}
                   & \multicolumn{3}{c}{\textbf{$L_2$ Error}}      \\ \hline
\textbf{Mechanism} & $\epsilon=0.2$ & $\epsilon=1$ & $\epsilon=5$ \\ \hline\hline
ExpM (G)           & 177.28         & 34.98        & 7.11          \\
EigM (G)           & 175.65         & 34.87        & 4.89          \\
DauM (G)           & 69.85          & 13.40        & 1.24          \\ \hline
GroupDP (G)            & 7394.67        & 1539.93      & 293.17        \\ \hline
\end{tabular}
\end{center}
\end{table}

Figure \ref{fig:a_util_pdiff} furthermore shows the error incurred by our mechanisms as $\Delta p$ varies in comparison to the Gaussian Mechanism from differential privacy, noting that differential privacy does not provide  distribution privacy guarantees.
DauM incurs  less error than  differential privacy when protecting the income property to an additive factor of $\Delta p_I\le 0.12$ (e.g., preventing an attacker from distinguishing between data with 44\% and 56\% of individuals with income $>\$50$K), or protecting the work class property to an additive factor of $\Delta p_W\le 0.28$. This demonstrates an advantage of our mechanisms over differential privacy, as they are able to  scale noise differently according to the specific sensitive property of interest and degree to which it needs to be protected. However, it also shows that  these mechanisms may require considerable amounts of noise to protect sensitive properties to larger values of $\Delta p$, suggesting the need for alternative or additional domain specific analysis.

Nonetheless, Figure \ref{fig:a_util_mech} shows how directional assumptions can reduce the error incurred by our mechanisms. This illustrates the potential to significantly reduce overall error by essentially only applying noise to query components with high correlation with the sensitive property being protected. On the other hand, our optimizations using adversarial uncertainty (i.e., EigM, DauM) were generally only noticeable for relatively large $\epsilon$.
These optimizations may be more applicable when there is less correlation between the query and sensitive properties, formally justifying a statistical release without additional noise if inherent adversarial uncertainty is sufficient to provide privacy.

\subsection{Property Inference Attack}
\label{sec:attack}
We now apply our mechanisms against a property inference attack aiming to determine whether the sensitive property has value
{$p_1=0.45$ or $p_2=0.55$}, i.e., $\Delta p= 0.1$.
Our attack is based on the standard meta-classifier technique~\citep{ateniese_hacking_2015,zhang_dataset-level_2020}, with logistic regression as the meta-classifier, the query $F$ as training features and the sensitive property $p$ as labels. We generate training features and labels using the auxiliary data to sample 200 shadow datasets, each of size $n=100$ records with half of the sampled datasets satisfying $p=p_1$ and half satisfying $p=p_2$.
To evaluate meta-classifier accuracy, we use 200 sampled subsets of size $n=100$ from the testing data, half with $p_1$ and half with $p_2$. We  report the average meta-classifier accuracy over 50 repetitions.

\begin{figure} 
    \centering
    \begin{subfigure}[t]{0.30\textwidth}
        \centering
        \includegraphics{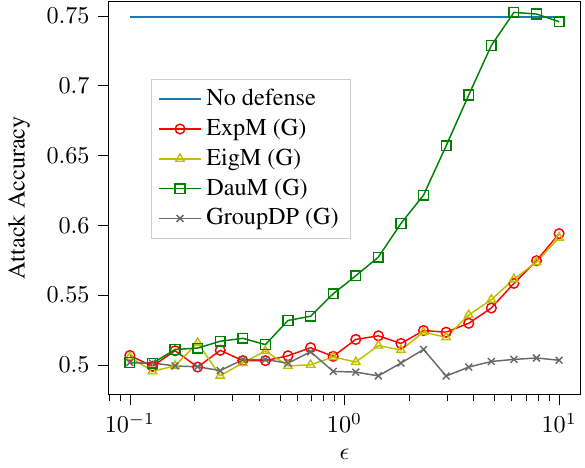}
        \vspace{-18pt}
    \end{subfigure}
    \hspace{32pt}
    \caption{
    Accuracy of attack determining whether the income property is $p_I=0.45$ or $p_I=0.55$.
    Defense mechanisms are applied with varying $\epsilon$ and $\delta=0.001$.
    }
    \label{fig:a_attack_accuracies_mech}
\end{figure}

For the income property $p_I$, the attack has 75\% success rate when statistics are released with no defense, while
Figure~\ref{fig:a_attack_accuracies_mech} shows that applying our mechanisms with $\epsilon=0.1$ reduces the attack accuracy to near $50\%$. Note that varying $\delta$ between 0.0001 and 0.01 had a much smaller effect on the attack accuracy than varying $\epsilon$ from 0.1 to 10, as we summarise and discuss in Appendix \ref{ap:more_attack}. Note also that the results for Laplace noise and for $p_W$ are omitted as they show similar trends.

Now, group differential privacy appears to provide the strongest defense out of all of the implemented mechanisms.
However, this is due to the excessive noise it applies, as was shown in Table \ref{tab:a_util}. In contrast, our mechanisms are more calibrated for $(\epsilon,\delta)$-distribution privacy with a trend of high to low defense as $\epsilon$ increases.

Finally, we remark that our mechanisms' empirical effectiveness against this attack does not necessarily guarantee defense capabilities against other attacks. Indeed, the need to defend against a range of attacks, including possibly not yet known attacks, highlights the importance of our mechanisms' formal guarantees. Nonetheless, these attack results provide an initial demonstration of our theoretical defenses translating to defense in practice.

\section{CONCLUDING REMARKS}

We presented an approach to protecting global properties of datasets using the distribution privacy framework and investigated two mechanisms for achieving distribution privacy --- the Wasserstein Mechanism, adapted from Pufferfish privacy, and the Expected Value Mechanism, generalizing the Gaussian Mechanism from attribute privacy.
Our experiments demonstrated that these mechanisms can reduce the accuracy of a property inference attack while providing significantly better utility than a group differential privacy baseline.

We believe that our work provides initial steps towards rigorously protecting global properties of datasets, leaving several open questions in this area. Future work could include generalizing our mechanisms to a wider range of  data distributions, designing new mechanisms under different data modeling assumptions, and evaluating mechanisms in a broader range of application scenarios.

\subsubsection*{Acknowledgements}
We thank the reviewers for their insightful  suggestions towards the improvement of this paper.

\bibliography{refs}

\appendix
\onecolumn

\section{ADVERSARIAL UNCERTAINTY RESULTS}
In this section, we provide proofs and further discussion for the results stated in Section \ref{sec:exp_gauss_variants}.
\label{ap:gauss_exp_proofs}
\gaussreducedmulti*
\begin{proof}
Since $f_{\theta}\sim\Gauss(\mu_{\theta}, \Sigma_{\theta})$ for each $\theta\in\Theta$, we can write each possible distribution of the query function in the form $F(\data)\given \theta=A_{\theta}Z+\mu_\theta$, where $Z$ is a standard normal vector and $A_{\theta}$ is an arbitrary matrix satisfying $A_{\theta}A_{\theta}^{\T}=\Sigma_{\theta}$. Then, for each $\theta\in\Theta$ and $x\in\R^m$,
\begin{align*}
    \Pr(F(\data)=x\given \theta)
    &=\Pr(A_\theta Z+\mu_\theta =x)\\
    &=\frac{1}{|A_{\theta}|}\Pr(Z=z_{\theta}),
\end{align*}
where $z_\theta=A_{\theta}^{-1}(x-\mu_{\theta})$.

Let $(\theta_i,\theta_j)\in\Psi$ be a pair of distributions. By assumption, we have $\Sigma_{\theta_i}=\Sigma_{\theta_j}$, thus, we may choose $A_{\theta_i}=A_{\theta_j}$ to be the same matrix. Then,
\begin{align*}
    \frac{\Pr(F(\data)=x\given \theta_i)}{\Pr(F(\data)=x\given \theta_j)}
    &=\frac{\frac{1}{|A_{\theta_i}|}\Pr(Z=z_{\theta_i})}{\frac{1}{|A_{\theta_j}|}\Pr(Z=z_{\theta_j})}\\
    &=\frac{\Pr(Z=z_{\theta_i})}{\Pr(Z=z_{\theta_i}+(z_{\theta_j}-z_{\theta_i}))}.
\end{align*}
Now, a standard normal distribution is spherically symmetric, so the distribution of $Z$ is independent of the orthonormal basis from which its components are drawn. Let us  fix an orthonormal basis $b_1,b_2,\dots,b_m$ such that $b_1$ is parallel to $z_{\theta_j}-z_{\theta_i}$. We may then express $Z$ in the form
\begin{align*}
    Z=\sum_{k=1}^{m}\lambda_k b_k,
\end{align*}
where $\lambda_k\sim\Gauss(0,1)$ for each $k$ are independently drawn. 
Writing $z_{\theta_i}=\sum_{k=1}^{m}\alpha_k b_k$, we then have
\begin{align*}
    \frac{\Pr\left(Z=z_{\theta_i}\right)}{\Pr\left(Z=z_{\theta_i}+(z_{\theta_j}-z_{\theta_i})\right)}
    &=\frac{\Pr\left(\lambda_1=\alpha_1\right)}{\Pr\left(\lambda_1=\alpha_1+\|z_{\theta_j}-z_{\theta_i}\|_2\right)}\prod_{k=2}^{m}\frac{\Pr\left(\lambda_k =\alpha_k\right)}{\Pr\left(\lambda_k =\alpha_k\right)}\\
    &=\exp\left(-\frac{1}{2}\left(\alpha_1^2-\left(\alpha_1+\|z_{\theta_j}-z_{\theta_i}\|_2\right)^2\right)\right)\\
    &=\exp\left(\alpha_1\|z_{\theta_j}-z_{\theta_i}\|_2+\frac{1}{2}\|z_{\theta_j}-z_{\theta_i}\|_2^2\right).
\end{align*}
The conditions of the theorem moreover give
\begin{align}
    (\mu_{\theta_i}-\mu_{\theta_j})^{\T} \Sigma_{\theta_i}^{-1}(\mu_{\theta_i}-\mu_{\theta_j})&\le (\epsilon/c)^2\nonumber\\
    \implies (\mu_{\theta_i}-\mu_{\theta_j})^{\T} (A_{\theta_i}A_{\theta_{i}}^{\T})^{-1}(\mu_{\theta_i}-\mu_{\theta_j})&\le (\epsilon/c)^2\nonumber\\
    \implies \|A_{\theta_{i}}^{-1}(\mu_{\theta_i}-\mu_{\theta_j})\|_2&\le \epsilon/c,\nonumber
\end{align}
thus, $\|z_{\theta_j}-z_{\theta_i}\|_2=\|A_{\theta_i}^{-1}(\mu_{\theta_j}-\mu_{\theta_i})\|_2\le \epsilon/c$. It follows that 
\begin{align*}
    \frac{\Pr\left(Z=z_{\theta_i}\right)}{\Pr\left(Z=z_{\theta_i}+(z_{\theta_j}-z_{\theta_i})\right)} \le \exp\left(\frac{|\alpha_1|\epsilon}{c} +\frac{\epsilon^2}{2c^2}\right).
\end{align*}
The above quantity is bounded by $\exp(\epsilon)$ whenever $|\alpha_1|\le c-\epsilon/2c$. We  now show that this occurs with probability at least $1-\delta$.

Since $\alpha_1\sim\Gauss(0,1)$, the standard Gaussian tail bound gives
\begin{align*}
    \Pr\left(|\alpha_1|>t\right)<\frac{\sqrt{2}\exp(-t^2/2)}{t\sqrt{\pi}}
\end{align*}
for $t>0$.  Letting $t=c-\epsilon/2c$, we would like to ensure that $\Pr(|\alpha_1|>t)<\delta$. It suffices to have
\begin{align*}
    \frac{\sqrt{2}\exp(-t^2/2)}{t\sqrt{\pi}}&<\delta\\
    \iff t\exp\left(\frac{t^2}{2}\right)&>\frac{\sqrt{2}}{\sqrt{\pi}\delta}\\
    \iff \ln(t)+\frac{t^2}{2}&>\ln\left(\frac{\sqrt{2}}{\sqrt{\pi}\delta}\right).
\end{align*}
We may assume that $\epsilon\le 1$ and $c=\sqrt{2\ln(1.25/\delta)}\ge 3/2$, so that the first term can be bounded by $\ln(t)=\ln(c-\epsilon/2c)\ge \ln(7/6)$. For the second term, $\epsilon\le 1$ implies that
\begin{align*}
    \frac{t^2}{2}
    =\frac{1}{2}\left(c-\frac{\epsilon}{2c}\right)^2
    =\frac{1}{2}\left(c^2-\epsilon+\frac{\epsilon^2}{4c^2}\right)\ge \frac{c^2-1}{2}.
\end{align*}
Thus, it suffices to have
\begin{align*}
    \ln\left(\frac{7}{6}\right)+\frac{c^2-1}{2}&\ge \ln\left(\frac{\sqrt{2}}{\sqrt{\pi}\delta}\right)\\
    \iff c^2 &\ge 2\ln\left(\frac{6\sqrt{2}\exp(1/2)}{7\sqrt{\pi}\delta}\right)
\end{align*}
which holds for $c=\sqrt{2\ln(1.25/\delta)}$ since $6\sqrt{2}\exp(1/2)/7\sqrt{\pi}<1.25$. Thus, $\Pr(|\alpha_1|>c-\epsilon/2c)<\delta$.

To conclude the proof, write each $x\in\R^m$ in the form $x=A_{\theta_i}\sum_{k=1}^{m}\alpha_k b_k+\mu_{\theta_i}$ and partition $\R^m$ as $\R^m=R_1\cup R_2$, where
\begin{align*}
    R_1&=\left\{x:|\alpha_1|\le c-\epsilon/2c\right\}\\
    R_2&=\left\{x:|\alpha_1|> c-\epsilon/2c\right\}.
\end{align*}
Then, $\Pr(F(\data)=x\given\theta_i)\le \exp(\epsilon)\Pr(F(\data)=x\given\theta_j)$ for all $x\in R_1$, and $\Pr(F(\data)\in R_2\given \theta_i)<\delta$. Fix an arbitrary subset $S\subseteq \R^m$, and define $S_1=S\cap R_1$ and $S_2=S\cap R_2$. We have
\begin{align*}
    \Pr(F(\data)\in S\given \theta_i)
    &=\Pr(F(\data)\in S_1\given \theta_i)+\Pr(F(\data)\in S_2\given \theta_i)\\
    &\le \exp(\epsilon)\Pr(F(\data)\in S_1\given \theta_j)+\delta\\
    &\le \exp(\epsilon)\Pr(F(\data)\in S\given \theta_j)+\delta.
\end{align*}
If follows that $\M$ satisfies  $(\epsilon,\delta)$-distribution privacy with respect to $\Psi$.
\end{proof}
Theorem \ref{thm:gauss_reduced_multi} gives a means of estimating the amount of privacy that can be guaranteed only using adversarial uncertainty in the distributions $f_\theta$ for each $\theta\in\Theta$. However, additional noise may be needed to ensure a desirable level of privacy, especially in cases where there may  not be much randomness inherent in the query. In such cases, we can add Gaussian noise to  manipulate the covariance matrices $\Sigma_{\theta}$ to satisfy the conditions required by the theorem.
\begin{corollary}
\label{cor:gauss_multi_reduced}
Suppose that $f_{\theta}\sim\Gauss(\mu_{\theta}, \Sigma_\theta)$  for each $\theta\in\Theta$ and that $\Sigma_{\theta_i}=\Sigma_{\theta_j}$ for each $(\theta_i,\theta_j)\in\Psi$. Let $c=\sqrt{2\ln(1.25/\delta)}$. Then, the mechanism $\M$ which outputs $F(\data)+Z$, where $Z\sim\Gauss(0,\Sigma)$, satisfies $(\epsilon,\delta)$-distribution privacy with respect to $\Psi$ as long as
\begin{align*}
    (\mu_{\theta_i}-\mu_{\theta_j})^{\T} (\Sigma_{\theta_i}+\Sigma)^{-1}(\mu_{\theta_i}-\mu_{\theta_j})\le (\epsilon/c)^2
\end{align*}
for all $(\theta_i,\theta_j)\in\Psi$.
\end{corollary}
\begin{proof}
Since $f_{\theta}\sim\Gauss(\mu_{\theta}, \Sigma_\theta)$ and $Z\sim\Gauss(0,\Sigma)$ are independent, we have that $F(\data)+Z\sim \Gauss(\mu_{\theta}, \Sigma_\theta+\Sigma)$ for each $\data\sim\theta\in\Theta$. 
The modified query $F(\data)+Z$ then satisfies the conditions of Theorem \ref{thm:gauss_reduced_multi} and the result follows.
\end{proof}
Alternatively, we can use the following set of sufficient conditions expressed in terms of the worst case $L_2$ distance between expected values, $\sensEt{\Psi}{F}$, defined  in Section \ref{sec:basic_exp}.
\begin{corollary}
\label{cor:gauss_eigen}
Suppose that $f_{\theta}\sim\Gauss(\mu_{\theta}, \Sigma_\theta)$  for each $\theta\in\Theta$ and that $\Sigma_{\theta_i}=\Sigma_{\theta_j}$ for each $(\theta_i,\theta_j)\in\Psi$. Let $c=\sqrt{2\ln(1.25/\delta)}$. Then, the mechanism $\M$ which outputs $F(\data)+Z$, where $Z\sim\Gauss(0,\Sigma)$, satisfies $(\epsilon,\delta)$-distribution privacy with respect to $\Psi$ as long as
\begin{align*}
    \Sigma_{\theta}+\Sigma \ge (c \sensEt{\Psi}{F}/\epsilon)^2I
\end{align*}
for all $\theta\in\Theta$. Equivalently, $\M$ satisfies  $(\epsilon,\delta)$-distribution privacy with respect to $\Psi$ as long as the minimum eigenvalue of $\Sigma_{\theta}+\Sigma$ is at least $(c \sensEt{\Psi}{F}/\epsilon)^2$ for all $\theta\in\Theta$.
\end{corollary}

\begin{proof}
We may assume that $\sensEt{\Psi}{F}>0$, so that we have $(c \sensEt{\Psi}{F}/\epsilon)^2I>0$. Then, using properties of positive definite matrices, we have
\begin{align*}
    \Sigma_{\theta}+\Sigma &\ge (c \sensEt{\Psi}{F}/\epsilon)^2I\\
    \implies (\Sigma_{\theta}+\Sigma)^{-1}&\le (\epsilon/c \sensEt{\Psi}{F})^2I.
\end{align*}
Let $(\theta_i,\theta_j)\in\Psi$. By definition of $\sensEt{\Psi}{F}$, we have $\|\mu_{\theta_i}-\mu_{\theta_j}\|_2\le \sensEt{\Psi}{F}$. Thus,
\begin{align*}
    (\mu_{\theta_i}-\mu_{\theta_j})^{\T} (\Sigma_{\theta_i}+\Sigma)^{-1}(\mu_{\theta_i}-\mu_{\theta_j})
    &\le (\mu_{\theta_i}-\mu_{\theta_j})^{\T}((\epsilon/c \sensEt{\Psi}{F})^2I)(\mu_{\theta_i}-\mu_{\theta_j})\\
    &=(\epsilon/c \sensEt{\Psi}{F})^2\|\mu_{\theta_i}-\mu_{\theta_j}\|_2^2\\
    &\le(\epsilon/c)^2.
\end{align*}
The conditions of Corollary \ref{cor:gauss_multi_reduced} are thus satisfied, and it follows that $\M$ satisfies $(\epsilon,\delta)$-distribution privacy with respect to $\Psi$.
\end{proof}
\subsection{Eigenvector Gaussian Mechanism}
\label{ap:eig_gauss}
\gausseig*
\begin{proof}
Let $\theta\in\Theta$. By assumption, $v_1,v_2,\dots, v_m$ are the normalized eigenvectors of $\Sigma_\theta$, and since $\Sigma_{\theta}$ is symmetric, these eigenvectors are orthogonal. We may thus write $\Sigma_\theta=\sum_{k}\lambda_{\theta,k}^{2}v_{k}v_{k}^{\T}$, where $\lambda_{\theta,k}^2=v_k^{\T}\Sigma_{\theta}v_k$ is the  eigenvalue of $\Sigma_\theta$ corresponding to the eigenvector $v_k$. Then, we have
\begin{align*}
    \Sigma_\theta+\Sigma&=\sum_{k}(\lambda_{\theta,k}^2+\sigma_k^2)v_{k}v_{k}^{\T}\\
    \implies (\Sigma_\theta+\Sigma)v_{r}&=\sum_{k}(\lambda_{\theta,k}^2+\sigma_k^2)v_{k}v_{k}^{\T}v_{r}\\
    &=(\lambda_{\theta,r}^2+\sigma_r^2)v_{r},
\end{align*}
for each $1\le r\le m$. 
The eigenvalues of $\Sigma_\theta+\Sigma$ are thus exactly the values $\lambda_{\theta,k}^2+\sigma_k^2$. Now, from our choice of $\sigma_k$,
\begin{align*}
    \lambda_{\theta,k}^2+\sigma_k^2\ge \lambda_{\theta,k}^2+\sigma_{\theta,k}^2\ge (c\sensEt{\Psi}{F}/\epsilon)^2,
\end{align*}
so that the minimum eigenvalue of $\Sigma_\theta+\Sigma$ is at least $(c\sensEt{\Psi}{F}/\epsilon)^2$. The conditions of Corollary \ref{cor:gauss_eigen} are thus  satisfied, and it follows that $\M$ satisfies $(\epsilon,\delta)$-distribution privacy with respect to $\Psi$.
\end{proof}

\subsection{Adversarial Uncertainty with Directional Assumptions}
\label{ap:dir_gauss}
We briefly discuss how directional assumptions (Section \ref{sec:dir_exp}) can be used in conjunction with adversarial uncertainty (Section \ref{sec:exp_gauss_variants}). Essentially, we apply noise in the direction of $\mu_{\theta_i}-\mu_{\theta_j}$ while also taking into account the variance of each distribution $f_\theta$. In doing so, we produce a variant of the Expected Value Mechanism which uses  directional  and adversarial uncertainty assumptions (DauM), described in Algorithm \ref{alg:Gaussian_Directional}. We present the mechanism for the case where the differences of mean vectors $\mu_{\theta_i}-\mu_{\theta_j}$ are all parallel to some vector $v$. As for the Directional Expected Value Mechanism (Section \ref{sec:dir_exp}) and Eigenvector Gaussian Mechanism (Section \ref{sec:exp_gauss_variants}), this mechanism can also be modified to account for cases where this assumption may not exactly hold.
\begin{algorithm}
\SetAlgoNoLine
Set  $c=\sqrt{2\ln(1.25/\delta)}$.\\
\For{each $\theta\in\Theta$}{
    Set $\mu_{\theta}=\E[F(\data)\given\data\sim\theta]$, $\Sigma_{\theta}=\Cov(F(\data)\given\data\sim\theta)$.
}
\For{each $(\theta_{i}, \theta_{j})\in\Psi$}{
    Set $\alpha_{i,j}=(\mu_{\theta_i}-\mu_{\theta_j})^{\T} v$.\\
    Find $\sigma_{i,j}^2>0$ such that $\Sigma_{\theta_i}+(\sigma_{i,j}^2-(\alpha_{i,j}c/\epsilon)^2)vv^{\T}>0$. \label{algline:dgm:7}
  }
Set $\sigma^2=\max_{(\theta_i,\theta_j)\in\Psi}\sigma_{i,j}^2$.\\
\Return{
$F(\data) + Z$
}, where $Z\sim\Gauss(0,\sigma^2 vv^{\T})$.
\caption{DauM (dataset $\data$, query $F$, distribution pairs $\Psi\subseteq\Theta\times\Theta$, privacy parameters $\epsilon$, $\delta$, unit vector $v$)}
\label{alg:Gaussian_Directional}
\end{algorithm}

To prove that Algorithm \ref{alg:Gaussian_Directional} satisfies distribution privacy, we will use the following lemma.
\begin{lemma}
\label{lem:sig_inv}
Let $v$ be a vector and let $\Sigma>0$. Suppose that $\Sigma v = \lambda v$. Then, $v^{\T}\Sigma^{-1}v=\lambda^{-1}\|v\|_2^2$.
\end{lemma}
\begin{proof}
We have
\begin{align*}
    \Sigma v = \lambda v
    &\implies \Sigma^{-1} v = \lambda^{-1}v\\
    &\implies v^{T}\Sigma^{-1} v = \lambda^{-1}v^{\T}v=\lambda^{-1}\|v\|_2^2.
\end{align*}
\end{proof}
\begin{theorem}
\label{thm:gauss_dir}
Suppose that $f_{\theta}\sim\Gauss(\mu_{\theta}, \Sigma_\theta)$  for each $\theta\in\Theta$ and that $\Sigma_{\theta_i}=\Sigma_{\theta_j}$ for each $(\theta_i,\theta_j)\in\Psi$. Suppose furthermore that, for each $(\theta_i,\theta_j)\in \Psi$, the vector $\mu_{\theta_i}-\mu_{\theta_j}$ is parallel to the unit vector $v$. Then, the mechanism $\M$ described in Algorithm \ref{alg:Gaussian_Directional} satisfies $(\epsilon,\delta)$-distribution privacy with respect to $\Psi$. 
\end{theorem}
\begin{proof}
Let $(\theta_i,\theta_j)\in\Psi$.
By Corollary \ref{cor:gauss_multi_reduced}, it suffices to show that
\begin{align}
    (\mu_{\theta_i}-\mu_{\theta_j})^{\T} (\Sigma_{\theta_i}+\sigma^2 vv^{\T})^{-1}(\mu_{\theta_i}-\mu_{\theta_j})\le (\epsilon/c)^2.
\label{eq:dir1}
\end{align}
By assumption, $\mu_{\theta_i}-\mu_{\theta_j}$ is parallel to $v$, hence  $\mu_{\theta_i}-\mu_{\theta_j}=\alpha_{i,j} v$, where $\alpha_{i,j}=(\mu_{\theta_i}-\mu_{\theta_j})^{\T} v$. Then,
\begin{align*}
    (\mu_{\theta_i}-\mu_{\theta_j})^{\T} (\Sigma_{\theta_i}+\sigma^2 vv^{\T})^{-1}(\mu_{\theta_i}-\mu_{\theta_j})
    =\alpha_{i,j}^2v^{\T}(\Sigma_{\theta_i}+\sigma^2 vv^{\T})^{-1}v.
\end{align*}
We would thus like to show that $v^{\T}(\Sigma_{\theta_i}+\sigma^2 vv^{\T})^{-1}v\le (\epsilon/c\alpha_{i,j})^2$. Now, from our choice of $\sigma^2$ and using properties of positive semi-definite matrices, we have
\begin{align*}
    \Sigma_{\theta_i}+\sigma^2 vv^{\T}\ge \Sigma_{\theta_i}+\sigma_{i,j}^2 vv^{\T}>(\alpha_{i,j}c/\epsilon)^2vv^{\T}.
\end{align*}
Thus, there exists some $\eta>0$ such that $\Sigma_{\theta_i}+\sigma^2 vv^{\T}-(\alpha_{i,j}c/\epsilon)^2vv^{\T}>\eta I$. We then have
\begin{align*}
    \Sigma_{\theta_i}+\sigma^2 vv^{\T}&>(\alpha_{i,j}c/\epsilon)^2vv^{\T}+\eta I\\
    \implies (\Sigma_{\theta_i}+\sigma^2 vv^{\T})^{-1}&<((\alpha_{i,j}c/\epsilon)^2vv^{\T}+\eta I)^{-1}
\end{align*}
Now, since $((\alpha_{i,j}c/\epsilon)^2vv^{\T}+\eta I)v=((\alpha_{i,j}c/\epsilon)^2+\eta)v$, Lemma \ref{lem:sig_inv} implies
\begin{align*}
    v^{\T}((\alpha_{i,j}c/\epsilon)^2vv^{\T}+\eta I)^{-1}v =((\alpha_{i,j}c/\epsilon)^2+\eta)^{-1}<(\epsilon/c\alpha_{i,j})^2.
\end{align*}
Thus, $v^{\T}(\Sigma_{\theta_i}+\sigma^2 vv^{\T})^{-1}v<v^{\T}((\alpha_{i,j}c/\epsilon)^2vv^{\T}+\eta I)^{-1}v<(\epsilon/c\alpha_{i,j})^2$, and  \eqref{eq:dir1}  holds. The conditions of Corollary \ref{cor:gauss_multi_reduced} are thus  satisfied, and it follows that $\M$ satisfies  $(\epsilon,\delta)$-distribution privacy with respect to $\Psi$.
\end{proof}
Note that the condition on line \ref{algline:dgm:7} of Algorithm \ref{alg:Gaussian_Directional} can always be satisfied by setting $\sigma_{i,j}^2=(\alpha_{i,j}c/\epsilon)^2$. However, doing so essentially disregards $\Sigma_{\theta_i}$ and reduces Algorithm \ref{alg:Gaussian_Directional} to the Gaussian variant of the Directional Expected Value Mechanism described in Section \ref{sec:dir_exp}. More generally a suitable and possibly smaller value of $\sigma_{i,j}^2$ can be found by solving $\det \left(\Sigma_{\theta_i}+(\sigma_{i,j}^2-(\alpha_{i,j}c/\epsilon)^2)vv^{\T}\right)=0$ then increasing $\sigma_{i,j}^2$ slightly to ensure that $\Sigma_{\theta_i}+(\sigma_{i,j}^2-(\alpha_{i,j}c/\epsilon)^2)vv^{\T}$ is strictly positive definite. The noise reduction provided by Algorithm \ref{alg:Gaussian_Directional} in comparison to the Directional Expected Value Mechanism is then proportional to  $(\alpha_{i,j}c/\epsilon)^2-\sigma_{i,j}^2$.
\section{ADDITIONAL PROOFS}
\subsection{Approximate Wasserstein Mechanism}
\label{ap:wass_proofs}
\boundedwhp*
\begin{proof}
For each $\theta\in\Theta$, let $R_{\theta}=\{t:\|t-\E[f_{\theta}]\|_1\le c\}$. Then, $\Pr(F(\data)\in R_{\theta}\given\theta)\ge 1-\delta/2$ for each $\theta\in\Theta$.

Fix a pair $(\theta_i,\theta_j)\in\Psi$, and consider the distributions $f_{\theta_i}$ and $f_{\theta_j}$. Let $\gamma\in\Gamma(f_{\theta_i},f_{\theta_j})$ be an arbitrary joint distribution with marginals $f_{\theta_i}$ and $f_{\theta_j}$, and define $R=R_{\theta_i}\times R_{\theta_j}$. Then, for all $(t,s)\in R$, the triangle inequality gives
\begin{align*}
    \|t-s\|_1
    &\le \|\E[f_{\theta_i}]-\E[f_{\theta_j}]\|_1+\|t-\E[f_{\theta_i}]\|_1+\|s-\E[f_{\theta_j}]\|_1\\
    &\le \sensE{\Psi}{F}+2c.
\end{align*}
By the union bound, we also have
\begin{align*}
    \int\int_{(t,s)\in R}\gamma(t,s)\dd{t}\dd{s}
    &=1-\int\int_{(t,s)\not\in R_{\theta_i}\times R_{\theta_j}}\gamma(t,s)\dd{t}\dd{s}\\
    &\ge 1-\int_{t\not\in R_{\theta_i}}f_{\theta_i}(t)\dd{t}-\int_{s\not\in R_{\theta_j}}f_{\theta_j}(s)\dd{s}\\
    &\ge 1- \delta/2-\delta/2\\
    &=1-\delta.
\end{align*}
It follows that $f_{\theta_i}$ and $f_{\theta_j}$ are $(\sensE{\Psi}{F}+2c,\delta)$-close.
\end{proof}
\subsection{Expected Value Mechanism}
\label{ap:exp_proofs}
The proofs of Theorem \ref{thm:exp_lap} and \ref{thm:exp_gauss} are straight-forward applications of the Laplace and Gaussian Mechanisms from differential privacy \citep{dwork_algorithmic_2014} using the condition that each pair of distributions $f_{\theta_i}$, $f_{\theta_j}$ are translations of each other. We provide only the proof of Theorem \ref{thm:exp_gauss}, as Theorem \ref{thm:exp_lap} can similarly be derived by replacing references to the Gaussian Mechanism from differential privacy with the Laplace Mechanism from differential privacy.
\expgauss*
\begin{proof}
Let $(\theta_i,\theta_j)\in\Psi$ be a pair of distributions and let $Z=(Z_1,Z_2,\dots, Z_m)$. Since $f_{\theta_i}$ and $f_{\theta_j}$ are translations of each other, there exists some $b\in\R^m$ such that $f_{\theta_i}(t)=f_{\theta_j}(t+b)$ for all $t\in\R^m$. Then, for all $S\subseteq \Range(\M)$,
\begin{align}
\Pr(\M(\data)\in S\given\theta_i)
&=\int_{t}f_{\theta_i}(t)\Pr(Z+t\in S)\dd{t}\nonumber\\
&=\int_{s}f_{\theta_j}(s)\Pr(Z+s-b\in S)\dd{s}.\nonumber
\end{align}
Now, note that $\|b\|_2=\|\E[f_{\theta_i}]- \E[f_{\theta_j}]\|_2\le \sensEt{\Psi}{F}$, and that $\sigma\ge  c\sensEt{\Psi}{F}/\epsilon$. Then, since each component of noise $Z_k\sim\Gauss(0,\sigma^2)$ is independently sampled, the Gaussian Mechanism from differential privacy \citep{dwork_algorithmic_2014} implies that 
\begin{align*}
\Pr(Z+s-b\in S)\le\exp(\epsilon)\Pr(Z+s\in S)+\delta
\end{align*}
for all $s\in\R^m$. 
Thus,
\begin{align*}
\Pr(\M(\data)\in S\given\theta_i)
&\le \int_{s}f_{\theta_j}(s)\left(\exp(\epsilon)\Pr(Z+s\in S)+\delta\right)\dd{s}\\
&= \exp(\epsilon)\int_{s}f_{\theta_j}(s)\Pr(Z+s\in S)\dd{s}+\delta\\
&=\exp(\epsilon)\Pr(\M(\data)\in S\given\theta_j)+\delta.
\end{align*}
It follows that $\M$ satisfies $(\epsilon,\delta)$-distribution privacy with respect to $\Psi$.
\end{proof}
\subsection{Directional Expected Value Mechanism}
\label{ap:dir_exp_proofs}
The proof of Theorem \ref{thm:dir_exp} applies ideas from the mechanism described in Section \ref{sec:basic_exp} while also taking into account the possible directions of the differences of mean vectors $\E[f_{\theta_i}]-\E[f_{\theta_j}]$. Note also that it can readily be modified to use Gaussian noise instead of Laplace noise.
\direxp*
\begin{proof}
Let $(\theta_i,\theta_j)\in\Psi$ be a pair of distributions. Since $f_{\theta_i}$ and $f_{\theta_j}$ are translations of each other, we have $f_{\theta_i}(t)=f_{\theta_j}(t+b)$, where $b=\E[f_{\theta_i}]-\E[f_{\theta_j}]$, for all $t\in\R^m$. Then, for all $S\subseteq \Range(\M)$,
\begin{align}
\Pr(\M(\data)\in S\given\theta_i)
&=\int_{t}f_{\theta_i}(t)\Pr(Yv+t\in S)\dd{t}\nonumber\\
&=\int_{s}f_{\theta_j}(s)\Pr(Yv+s-b\in S)\dd{s}.
\label{eq:dir_exp1}
\end{align}
By assumption, $b=\E[f_{\theta_i}]- \E[f_{\theta_j}]$ is parallel to the unit vector $v$. Thus, letting $S'=S\cap (s+\Span(v))$, we  have
\begin{align*}
    \Pr(Yv+s-b\in S)
    &=\Pr(Yv+s-b\in S')\\
    &=\Pr(Yv-b\in S'-s)\\
    &=\Pr(Y-\|b\|_2\in v^{\T}(S'-s)),
\end{align*}
where $v^{\T}(S'-s)=\{v^{\T}(x-s):x\in S'\}$.

Now, observe that $\|\|b\|_2\|_1=\|b\|_2=\|\E[f_{\theta_i}]- \E[f_{\theta_j}]\|_2\le \sensEt{\Psi}{F}$. Since $Y\sim \Lap(\sensEt{\Psi}{F}/\epsilon)$, the Laplace Mechanism from differential privacy \citep{dwork_calibrating_2006}
then gives that
\begin{align*}
\Pr(Y-\|b\|_2\in v^{\T}(S'-s))
&\le\exp(\epsilon)\Pr(Y\in v^{\T}(S'-s))\\
&=\exp(\epsilon)\Pr(Yv+s\in S).
\end{align*}
Thus, $\Pr(Yv+s-b\in S)\le \exp(\epsilon)\Pr(Yv+s\in S)$
for all $s\in\R^m$.
Together with \eqref{eq:dir_exp1}, this implies
\begin{align*}
\Pr(\M(\data)\in S\given\theta_i)
&\le \exp(\epsilon)\int_{s}f_{\theta_j}(s)\Pr(Yv+s\in S)\dd{s}\\
&=\exp(\epsilon)\Pr(\M(\data)\in S\given\theta_j).
\end{align*}
It follows that $\M$ satisfies $(\epsilon,0)$-distribution privacy with respect to $\Psi$.
\end{proof}

\subsection{Approximations with Max-Divergence}
\label{ap:approx_proofs}

\othermaxdiva*
\begin{proof}
For convenience of presentation, we write $F=F(\data)$ and  refer interchangeably between a random variable and its underlying distribution. 
By assumption, $\M(\data)\given F$ is independent of $\data$, thus, $\widetilde{\M}(\data)\given F$ is also independent of $\data$.
Let us write $\widetilde{\M}_{f}$ to denote the distribution of $\widetilde{\M}(\data)\given F\sim f$.

Since $\M$ satisfies $(\epsilon,\delta)$-distribution privacy with respect to $\Psi$, it follows that $\widetilde{\M}$ will as well when the approximations $\tilde{f}_{\theta}$ are  the true distributions. Thus, for all pairs $(\theta_i,\theta_j)\in\Psi$ and subsets $S\subseteq \Range(\widetilde{\M})$,
\begin{align}
    \Pr(\widetilde{\M}_{\tilde{f}_{\theta_i}}\in S)\le \exp(\epsilon)\Pr(\widetilde{\M}_{\tilde{f}_{\theta_j}}\in S)+\delta.
    \label{eq:other_maxdiv1}
\end{align}
Since $\max\left(\maxdiva{\eta}{\tilde{f}_{\theta}}{f_{\theta}},\maxdiva{\eta}{f_{\theta}}{\tilde{f}_{\theta}}\right)\le \lambda$
for all $\theta\in\Theta$, we also have
\begin{align*}
    \Pr(\widetilde{\M}_{f_{\theta_i}}\in S)
    &=\int_{t}f_{\theta_i}(t)\Pr(\widetilde{\M}(\data)\in S\given F=t)\dd{t}\\
    &\le \int_{t}(\exp(\lambda)\tilde{f}_{\theta_i}(t)+\eta)\Pr(\widetilde{\M}(\data)\in S\given F=t)\dd{t}\\
    &\le \exp(\lambda)\int_{t}\tilde{f}_{\theta_i}(t)\Pr(\widetilde{\M}(\data)\in S\given F=t)\dd{t}+\eta\\
    &=\exp(\lambda)\Pr(\widetilde{\M}_{\tilde{f}_{\theta_i}}\in S)+\eta.
\end{align*}
Similarly, $\Pr(\widetilde{\M}_{\tilde{f}_{\theta_j}}\in S)\le \exp(\lambda)\Pr(\widetilde{\M}_{f_{\theta_j}}\in S)+\eta$. Together with \eqref{eq:other_maxdiv1}, this gives
\begin{align*}
    \Pr(\widetilde{\M}(\data)\in S\given \theta_{i})
    &=\Pr(\widetilde{\M}_{f_{\theta_i}}\in S)\\
    &\le \exp(\lambda)\Pr(\widetilde{\M}_{\tilde{f}_{\theta_i}}\in S)+\eta\\
    &\le \exp(\lambda)(\exp(\epsilon)\Pr(\widetilde{\M}_{\tilde{f}_{\theta_j}}\in S)+\delta)+\eta\\
    &\le \exp(\lambda)(\exp(\epsilon)(\exp(\lambda)\Pr(\widetilde{\M}_{f_{\theta_j}}\in S)+\eta)+\delta)+\eta\\
    &=\exp(\epsilon')\Pr(\widetilde{\M}_{f_{\theta_j}}\in S)+\delta'\\
    &=\exp(\epsilon')\Pr(\widetilde{\M}(\data)\in S\given \theta_{j})+\delta'.
\end{align*}
It follows that $\widetilde{\M}$ satisfies $(\epsilon',\delta')$-distribution privacy with respect to $\Psi$.
\end{proof}
\subsection{Approximations with Wasserstein Distance}
To prove Theorem \ref{thm:relax_was_lap}, we first observe that the Wasserstein Mechanism of Section  \ref{sec:Wass} can be viewed as a result on the effect of Laplace noise on the max-divergence between distributions. In particular, Theorem \ref{thm:wasserstein} may be reinterpreted as follows.
\begin{lemma}
\label{lem:was_maxdiv_lap}
Let $F$ and $\widetilde{F}$ be two distributions on $\R^m$. Let $Z=(Z_1,Z_2,\dots, Z_m)$, where $Z_k\sim\Lap(\Winf{F}{\widetilde{F}}/\epsilon)$ for each $k$. Then, $\maxdiv{F+Z}{\widetilde{F}+Z}\le \epsilon$. 
\end{lemma}

We can then use Lemma \ref{lem:was_maxdiv_lap} to prove Theorem \ref{thm:relax_was_lap}. 

\relaxwaslap*
\begin{proof}
For convenience of presentation, we write $F=F(\data)$ and  refer interchangeably between a random variable and its underlying distribution. 
Let us write $\widetilde{\M}_{f}$ to denote the distribution of $\widetilde{\M}(\data)\given F\sim f$.

Since $\M$ satisfies $(\epsilon,\delta)$-distribution privacy with respect to $\Psi$, it follows that $\widetilde{\M}$ will as well when the approximations $\tilde{f}_{\theta}$ are  the true distributions. Thus, for all pairs $(\theta_i,\theta_j)\in\Psi$ and subsets $S\subseteq \Range(\widetilde{\M})$,
\begin{align*}
    \Pr(\widetilde{\M}_{\tilde{f}_{\theta_i}}\in S)\le \exp(\epsilon)\Pr(\widetilde{\M}_{\tilde{f}_{\theta_j}}\in S)+\delta.
\end{align*}
In turn, this implies
\begin{align}
\Pr(\widetilde{\M}_{\tilde{f}_{\theta_i}}+Z'\in S)\nonumber
&=\int_{z}\Pr(Z'=z)\Pr(\widetilde{\M}_{\tilde{f}_{\theta_i}}\in S-z)\dd{z}\nonumber\\
&\le \int_{z}\Pr(Z'=z)(\exp(\epsilon)\Pr(\widetilde{\M}_{\tilde{f}_{\theta_i}}\in S-z)+\delta)\dd{z}\nonumber\\
&\le \exp(\epsilon)\Pr(\widetilde{\M}_{\tilde{f}_{\theta_j}}+Z'\in S)+\delta.
\label{eq:relax_was_lap1}
\end{align}
Now, using properties of $\infty$-Wasserstein distance defined using the $L_1$ norm, we have
\begin{align*}
    \Winf{\widetilde{\M}_{f_{\theta_i}}}{\widetilde{\M}_{\tilde{f}_{\theta_i}}}
    &=\Winf{F+Z\given F\sim f_{\theta_i}}{F+Z\given F\sim \tilde{f}_{\theta_i}}\\
    &=\Winf{F\given F\sim f_{\theta_i}}{F\given F\sim \tilde{f}_{\theta_i}}\\
    &=\Winf{f_{\theta_i}}{\tilde{f}_{\theta_i}}\\
    &\le W.
\end{align*}
Hence, by Lemma \ref{lem:was_maxdiv_lap}, we have $\maxdiv{\widetilde{\M}_{f_{\theta_i}}+Z'}{\widetilde{\M}_{\tilde{f}_{\theta_i}}+Z'}\le \lambda$. Similarly, $\maxdiv{\widetilde{\M}_{\tilde{f}_{\theta_j}}+Z'}{\widetilde{\M}_{f_{\theta_i}}+Z'}\le \lambda$. Together with \eqref{eq:relax_was_lap1}, this gives
\begin{align*}
    \Pr(\widetilde{\M}(\data)+Z'\in S\given \theta_{i})
    &=\Pr(\widetilde{\M}_{f_{\theta_i}}+Z'\in S)\\
    &\le \exp(\lambda)\Pr(\widetilde{\M}_{\tilde{f}_{\theta_i}}+Z'\in S)\\
    &\le \exp(\lambda)(\exp(\epsilon)\Pr(\widetilde{\M}_{\tilde{f}_{\theta_j}}+Z'\in S)+\delta)\\
    &\le \exp(\lambda)(\exp(\epsilon)(\exp(\lambda)\Pr(\widetilde{\M}_{f_{\theta_j}}+Z'\in S))+\delta)\\
    &=\exp(\epsilon')\Pr(\widetilde{\M}_{f_{\theta_j}}+Z'\in S)+\delta'\\
    &=\exp(\epsilon')\Pr(\widetilde{\M}(\data)+Z'\in S\given \theta_{j})+\delta'.
\end{align*}
It follows that $\widetilde{\M}+Z'$ satisfies $(\epsilon',\delta')$-distribution privacy with respect to $\Psi$.
\end{proof}
We remark that the privacy guarantees in Theorem \ref{thm:relax_was_lap} may sometimes be attainable even without applying the additional noise $Z'$. Observe that the proof of Theorem \ref{thm:relax_was_lap} only requires $Z'$ to satisfy $\maxdiva{\eta}{\M_{f_{\theta}}+Z'}{\M_{\tilde{f}_{\theta}}+Z'}\le\lambda$ and $\maxdiva{\eta}{\M_{\tilde{f}_{\theta}}+Z'}{\M_{f_{\theta}}+Z'}\le\lambda$. Thus, if these inequalities are satisfied for $Z'=0$, then the Expected Value Mechanism may be applied safely using the approximations $\tilde{f}_{\theta}$ with no additional noise required.

\section{ADDITIONAL EXPERIMENTAL DETAILS}
\subsection{Gaussian Modeling}
\label{ap:more_results_modeling}
To approximate each distribution $f_{\theta}$ with a Gaussian distribution, we used the standard parameter estimation approach \cite{forbes_statistical_2011} of randomly sampling from $f_{\theta}$, then computing the sample mean $\mu_{\theta}$ and covariance matrix $\Sigma_{\theta}$ to yield an approximation $\Gauss(\mu_{\theta},\Sigma_{\theta})$. Specifically, for each distribution $\theta\in\Theta$,  we randomly sampled a subset of 100 records from the Adult dataset having the value of the global sensitive property $p_I$ or $p_W$ as specified by $\theta$. We then computed the query $F$ over this subset, producing one sample value for $f_{\theta}$. Repeating this process \num{1000} times yielded \num{1000} samples from which we computed a mean vector $\mu_{\theta}$ and covariance matrix $\Sigma_{\theta}$.

\begin{figure}
    \centering
     \begin{subfigure}[t]{0.46\textwidth}
        \centering
        \includegraphics{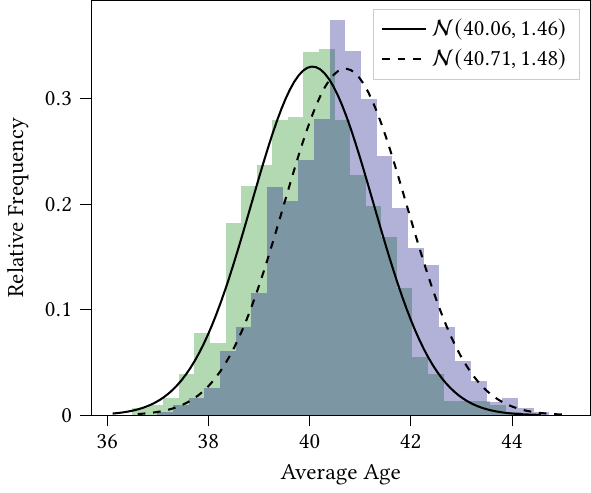}
        \label{fig:a_model_approx_a}
    \end{subfigure}
    \begin{subfigure}[t]{0.46\textwidth}
        \centering
        \includegraphics{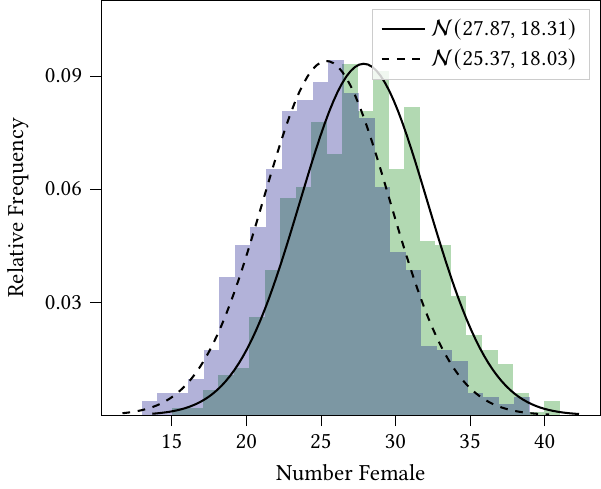}
        \label{fig:a_model_approx_c}
    \end{subfigure}
    \caption{
     Gaussian approximations for selected statistics of interest over randomly sampled subsets of the Adult dataset. Here, the subsets are sampled to have $n=100$ records with  fixed proportions $p_I$ of individuals having income $>\$50$K,
     where green and purple histograms are
    used for $p_I = 0.45$ and $p_I = 0.55$, respectively.}
    \label{fig:a_model_approx}
\end{figure}

Figure \ref{fig:a_model_approx} provides a visualization of our Gaussian approximations for selected query components and underlying sensitive property values. We remark that the variance of each component does vary slightly as the underlying sensitive property value varies, suggesting that our assumption on the covariance matrices $\Sigma_{\theta_i}$ and $\Sigma_{\theta_j}$ being equal for all pairs $(\theta_i,\theta_j)\in\Psi$ cannot hold exactly. For instance, the variance of the number of females in each subset of data is $\approx 18.31$ when $p_I=0.45$, and $\approx 18.03$ when $p_I=0.55$. However, the variance and covariance values are similar enough that we may still apply our mechanisms, albeit with the caveat of slightly weaker privacy guarantees as discussed in Section \ref{sec:relax}.

\subsection{Further Attack Results}
\label{ap:more_attack}
Table \ref{tab:a_gauss_delta} shows how the accuracy of the attack in Section \ref{sec:attack} changes as we apply Gaussian variants of the Expected Value Mechanism with varying values of $\epsilon$ and $\delta$. We note that varying $\delta$ between 0.0001 and 0.01 does not substantially impact the measured attack accuracies in comparison to varying $\epsilon$ between 0.2 and 5. This is expected as our Gaussian variants of the Expected Value Mechanism roughly apply noise scaled to $c\sensEt{\Psi}{F}/\epsilon$, where $c=\sqrt{2\ln (1.25/\delta)}$ grows sub-logarithmically as $\delta$ decreases. Nonetheless, we remark that there is still an important theoretical distinction between $(\epsilon,\delta)$-distribution privacy for different values of $\delta$, even if not evident against this specific property inference attack.

\begin{table}[t]
\caption{Accuracy of the attack on income when Gaussian variants of the Expected Value Mechanism are applied with varying values of $\epsilon$ and $\delta$. Here,  $\Delta p_{I}=0.1$. That is, the attack aims to distinguish between subsets of data satisfying $p_I=0.45$ and subsets satisfying $p_I=0.55$.} \label{tab:a_gauss_delta}
\begin{center}
\begin{tabular}{llccc}
\cline{3-5}
\textbf{}                 & \multicolumn{1}{c}{\textbf{}} & \multicolumn{3}{c}{\textbf{Attack Accuracy}}     \\ \hline 
\multicolumn{2}{c}{\textbf{Mechanism}}                    & $\delta=0.0001$ & $\delta=0.001$ & $\delta=0.01$ \\ \hline \hline
\multirow{3}{*}{ExpM (G)} & $\epsilon=0.2$                & 0.492 & 0.500 & 0.502 \\
                          & $\epsilon=1$
& 0.506 & 0.511 & 0.520 \\
                          & $\epsilon=5$
& 0.530 & 0.539 & 0.549 \\ \hline
\multirow{3}{*}{EigM (G)} & $\epsilon=0.2$                & 0.507 & 0.501 & 0.512 \\
                          & $\epsilon=1$
& 0.510 & 0.512 & 0.503 \\
                          & $\epsilon=5$
& 0.537 & 0.550 & 0.545 \\ \hline
\multirow{3}{*}{DauM (G)} & $\epsilon=0.2$                & 0.517 & 0.508 & 0.511 \\
                          & $\epsilon=1$
& 0.548 & 0.545 & 0.562 \\
                          & $\epsilon=5$
& 0.714 & 0.739 & 0.744 \\ \hline
\end{tabular}
\end{center}
\vspace{-10pt}
\end{table}

\vfill

\end{document}